\documentclass[a4paper,12pt,twoside]{article}

\usepackage{latexsym,epic,eepic,float}
\usepackage{color,colordvi}
\usepackage{shadow}
\usepackage{amsmath,amssymb}
\usepackage{epsfig}
\usepackage{lscape}
\usepackage{rotating}
\usepackage{subfigure}
\usepackage{multicol}
\usepackage{calc}

\newlength{\maxmin}
\usepackage[T1]{fontenc}
\usepackage[utf8]{inputenc}
\usepackage{authblk}

\newtheorem{theorem}{Theorem}[section]
\newtheorem{theo}[theorem]{Theorem}
\newtheorem{lem}[theorem]{Lemma}

\newtheorem{cor}[theorem]{Corollary}
\newtheorem{prop}[theorem]{Proposition}
\newtheorem{Rem}[theorem]{Remark}
\newtheorem{Ex}[theorem]{Example}

\newcommand{\exit}{{\mbox{\, \vspace{3mm}}} \hfill\mbox{$\square$}}

\rm 

\numberwithin{equation}{section}

\title{Generalized Phase-Type Distribution and Competing Risks for Markov Mixtures Process}

\author{B.A. Surya\footnote{Email address: budhi.surya@vuw.ac.nz; Postal address: School of Mathematics and Statistics, Victoria University of Wellington, Gate 6 Kelburn PDE, Wellington 6140, New Zealand.}\\ Victoria University of Wellington \\ School of Mathematics and Statistics \\ Wellington, New Zealand }

\date{11 November 2016}

\begin{document}
\maketitle \pagestyle{myheadings} \markboth{B.A. Surya}{Generalized Phase-Type Distribution and Competing Risks}

\begin{abstract}
Phase-type distribution has been an important probabilistic tool in the analysis of complex stochastic system evolution. It was introduced by Neuts \cite{Neuts1975} in 1975. The model describes the lifetime distribution of a finite-state absorbing Markov chains, and has found many applications in wide range of areas. It was brought to survival analysis by Aalen \cite{Aalen1995} in 1995. However, the model has lacks of ability in modeling heterogeneity and inclusion of past information which is due to the Markov property of the underlying process that forms the model. We attempt to generalize the distribution by replacing the underlying by Markov mixtures process. Markov mixtures process was used to model jobs mobility by Blumen \cite{Blumen} et al. in 1955. It was known as the mover-stayer model describing low-productivity workers tendency to move out of their jobs by a Markov chains, while those with high-productivity tend to stay in the job. Frydman \cite{Frydman2005} later extended the model to a mixtures of finite-state Markov chains moving at different speeds on the same state space. In general the mixtures process does not have Markov property. We revisit the mixtures model  \cite{Frydman2005} for mixtures of multi absorbing states Markov chains, and propose generalization of the phase-type distribution under competing risks. The new distribution has two main appealing features: it has the ability to model heterogeneity and to include past information of the underlying process, and it comes in a closed form. Built upon the new distribution, we propose conditional forward intensity which can be used to determine rate of occurrence of future events (caused by certain type) based on available information. Numerical study suggests that the new distribution and its forward intensity offer significant improvements over the existing model.

\medskip

\textit{MSC2010 subject classifications}: 60J20, 60J27, 60J28, 62N99


\end{abstract}

\section{Introduction}

Phase-type distribution is known to be a dense class of distributions, which can approximate any distribution arbitrarily well. It was introduced in 1975 by Neuts (\cite{Neuts1975}, \cite{Neuts1981}) and has found many applications in wide range of fields ever since.

When jumps distribution of compound Poisson process is modeled by phase-type distributions, it results in a dense class of L\'evy processes (Asmussen \cite{Asmussen2003}). The advantage of working under phase-type distribution is that it allows us to get some analytically tractable results in applications: like e.g. in actuarial science (Rolski et al. \cite{Rolski}, Albrecher and Asmussen \cite{Asmussen2010}, Lin and Liu \cite{Lin2007}, Lee and Lin \cite{Lee}), option pricing (Asmussen et al. \cite{Asmussen2004}, Rolski et. al \cite{Rolski}),  queueing theory (Badila et al. \cite{Badila}, Chakravarthy and Neuts \cite{Chakravarthy}, Buchholz et al. \cite{Buchholz}, Breuer and Baum \cite{Breuer}, Asmussen \cite{Asmussen2003}),  reliability theory (Assaf and Levikson \cite{Assaf1982}, Okamura and Dohi \cite{Okamura}), and in survival analysis (Aalen \cite{Aalen1995}, Aalen and Gjessing \cite{Aalen2001}).

The phase-type distribution $\overline{F}$ is expressed in terms of a Markov jump process $(X_t)_{t\geq 0}$ with finite state space $\mathbb{S}=E\cup \{\Delta\}$, where for some integer $m\geq 1$, $E=\{i: i=1,...,m\}$ is non absorbing states and $\Delta$ is the absorbing state, and an initial distribution $\boldsymbol{\pi}$, such that $\overline{F}$ is the distribution of time until absorption,
\begin{equation}\label{eq:DefTime}
\tau=\inf\{t\geq 0: X_t=\Delta\} \quad \mathrm{and} \quad \overline{F}(t)=\mathbb{P}\{\tau> t\}.
\end{equation}

Unless stated otherwise, we assume for simplicity that the initial probability $\boldsymbol{\pi}$ has zero mass on the absorbing state $\Delta$, so that $\mathbb{P}\{\tau>0\}=1$. We also refer to $\Delta$ as the m+1th element of state space $\mathbb{S}$, i.e., $\Delta=m+1$. Associated with the time propagation of $X$ on the state space $\mathbb{S}$ is the intensity matrix $\mathbf{Q}$. This matrix is partitioned according to the process moving to non-absorbing states $E$ and to single absorbing state $\Delta$, and has the block-partitioned matrix form:
\begin{equation}\label{eq:MatQ}
\mathbf{Q} = \left(\begin{array}{cc}
  \mathbf{T} & \boldsymbol{\delta} \\
  \mathbf{0} & 0 \\
\end{array}\right),
\end{equation}
where $\mathbf{T}$ is $m\times m-$ nonsingular matrix which together with exit vector $\boldsymbol{\delta}$ satisfy
\begin{equation}\label{eq:Tdone}
\mathbf{T}\mathbf{1}+\boldsymbol{\delta}=\mathbf{0},
\end{equation}
with $\mathbf{1}=(1,...,1)^{\top}$, as the rows of the intensity matrix $\mathbf{Q}$ sums to zero. That is to say that the entry $q_{ij}$ of intensity matrix $\mathbf{Q}$ satisfies the following properties:
\begin{equation}\label{eq:matq}
q_{ii}\leq 0, \; \; q_{ij}\geq 0, \; \; \sum_{j\neq i} q_{ij}=-q_{ii}=q_i, \quad (i,j)\in \mathbb{S}.
\end{equation}
As $\boldsymbol{\delta}$ is a non-negative vector, (\ref{eq:Tdone})-(\ref{eq:matq}) implies that $\mathbf{T}$ to be a negative definite matrix, i.e., $\mathbf{1}^{\top}\mathbf{T}\mathbf{1}<0$. The matrix $\mathbf{T}$ is known as the phase generator matrix of $X$. The absorption is certain if and only if $\mathbf{T}$ is nonsingular, see \cite{Neuts1981}, \cite{Asmussen2003}, \cite{Breuer}.

Following Theorem 3.4 and Corollary 3.5 in \cite{Asmussen2003} and by the homogeneity of $X$, the transition probability matrix $\mathbf{P}(t)$ of $X$ over the period of time $(0,t)$ is
\begin{equation}\label{eq:transprob}
\mathbf{P}(t)= \exp(\mathbf{Q} t), \quad t\geq 0,
\end{equation}
where $\exp(\mathbf{Q} t)$ is the $(m+1)\times (m+1)$ matrix exponential defined by the series
\begin{equation}\label{eq:expm}
\exp(\mathbf{Q} t) =\sum_{n=0}^{\infty} \frac{1}{n!} (\mathbf{Q}t)^n.
\end{equation}

The entry $q_{ij}$ has probabilistic interpretation: $1/(-q_{ii})$ is the expected length of time that $X$ remains in state $i\in E$, and $q_{ij}/q_{i}$ is the probability that when a transition out of state $i$ occurs, it is to state $j\in\mathbb{S}$, $i\neq j$. The representation of the distribution $\overline{F}$ is uniquely specified by $(\boldsymbol{\pi},\mathbf{T})$. We refer among others to Neuts \cite{Neuts1981}, Asmussen \cite{Asmussen2003} and Rolski et al. \cite{Rolski} for more details. Following Theorem 8.2.3 and Theorem 8.2.5 in \cite{Rolski} (or Proposition 4.1 in \cite{Asmussen2003}), we have
\begin{equation}\label{eq:DistrDefTime}
 \overline{F}(t)=\boldsymbol{\pi}^{\top} e^{\mathbf{T} t} \mathbf{1} \quad \mathrm{and} \quad f(t)=\boldsymbol{\pi}^{\top} e^{\mathbf{T} t} \boldsymbol{\delta}.
 \end{equation}

In survival and event history analysis, e.g., Aalen \cite{Aalen1995}, Aalen and Gjessing \cite{Aalen2001} and Aalen et al. \cite{Aalen2008}, intensity function $\alpha(t)$, for $t>0$, has been used to determine rate of occurrence of an event in time interval $(t,t+dt]$. It is defined by
\begin{equation}\label{eq:alpha}
\alpha(t)dt=\mathbb{P}\{\tau\leq t+dt \vert \tau>t\}.
\end{equation}
Solution to this equation is given in terms of the distributions (\ref{eq:DistrDefTime}) by
\begin{equation}\label{eq:alphasol}
\alpha(t)=\frac{f(t)}{\overline{F}(t)}=\frac{ \boldsymbol{\pi}^{\top} e^{\mathbf{T} t} \boldsymbol{\delta}}{\boldsymbol{\pi}^{\top} e^{\mathbf{T} t} \mathbf{1}}.
\end{equation}
This intensity has been used in greater applications in biostatistics (\cite{Aalen1995}, \cite{Aalen2008}).

Recall that the underlying process that forms the phase-type distribution (\ref{eq:DistrDefTime}) is based on continuous-time finite-state absorbing Markov chains. We attempt to generalize the distribution for non Markovian stochastic process.

Quite recently, Frydman \cite{Frydman2005} discussed mixture model of Markov chains moving at different speed as a generalization of the mover-stayer model introduced by Blumen et al. \cite{Blumen}. The result was then applied to study the dynamics of credit ratings, see Frydman and Schuermann \cite{Frydman2008}. As shown in \cite{Frydman2008}, the mixture process does not have Markov property. We revisit the model by considering a mixture of two absorbing Markov chains moving with different speed (intensity matrix). We derive conditional distribution of the time until absorption of the mixture process and propose generalizations of (\ref{eq:DistrDefTime}). Due to non Markovian nature of the mixture process, the generalized distribution allows for the inclusion of available past information of the mixture process and modeling heterogeneity. Built upon the new distribution, we propose conditional forward intensity which extends further (\ref{eq:alpha}) to determine rate of occurrence of future events based on available information about the past of the underlying process. The new distribution and its intensity have the ability to capture path dependence and heterogeneity.

The organization of this paper is as follows. We discuss in Section 2 the mixture process for further details. Section 3 presents some preliminary results, which are derived based on extending the results in \cite{Frydman2008}, in particular to continuous time setting. The main results concerning the new phase-type distributions and the forward intensity are given in Section 4. We also discuss in Section 4 residual lifetime and occupation time for the mixture process. We further extend the model in Section 5 for multi absorbing states (competing risks) and discuss conditional joint distribution of the lifetime of the mixtures process and the type-specific events (risk). Subsequently, we derive the corresponding type-specific conditional forward intensity for the competing risks.  Section 6 discusses numerical examples on married-and-divorced problem to compare the performance of existing phase-type distributions and intensity with the new ones proposed. Numerical study suggests that the new distributions and forward intensity offer significant improvements and reveal resemblances of data found e.g. in Aalen et al. \cite{Aalen2008} (p.222) and the common U-bend shape of lifetime found in the 2010 \textit{The Economist} article. Section 7 summarizes this paper with some conclusions.

\section{Markov mixtures processes}
In this paper, we consider $(X_t)_{t\geq0}$ as a mixture of two Markov jump processes $(X_t^Q)_{t\geq 0}$ and $(X_t^G)_{t\geq 0}$ both of which propagate on the same state space $\mathbb{S}=E\cup\{\Delta\}$. The corresponding intensity matrix $\mathbf{G}$ for the process $X^G$ is given by
\begin{equation}\label{eq:GH}
\mathbf{G} = \left(\begin{array}{cc}
  \boldsymbol{\Psi}\mathbf{T} & \boldsymbol{\Psi}\boldsymbol{\delta} \\
  \mathbf{0} & 0 \\
\end{array}\right),
\end{equation}
where $\boldsymbol{\Psi}$ is $m\times m$ nonsingular matrix, whilst the $m\times m$ matrix $\mathbf{T}$ and $m\times 1$ vector $\boldsymbol{\delta}$ are the elements of the intensity matrix $\mathbf{Q}$ (\ref{eq:MatQ}) of $X^Q$ satisfying the constraint (\ref{eq:Tdone}). As before, we refer $i=m+1$ to the absorbing state $\Delta$.

Frydman and Schuermann \cite{Frydman2008} proposed the use of mixture models for the estimation in discrete time of credit rating dynamics. In their work, intensity matrix $\mathbf{G}$ was taken in a simpler form where $\boldsymbol{\Psi}$ was given by diagonal matrix
\begin{equation}\label{eq:Lambda}
\boldsymbol{\Psi}=\textrm{diag}(\psi_1,\psi_2,...,\psi_{m}).
\end{equation}
Furthermore, they considered the absorbing state $\Delta$ corresponding to the censored events such as bond issuing firms being default, the bond are withdrawn from being rated, expiration of the bond, or calling back of the bond, etc.

Markov mixture process is a generalization of mover-stayer model, a mixture of two discrete-time Markov chains which was introduced by Blumen et al \cite{Blumen} in 1955 to model population heterogeneity in jobs mobility. In the mover-stayer model \cite{Blumen}, the population of workers consists of stayers (workers who always stay in the same job category ($\psi_i=0$)) and movers (workers who move according to a stationary Markov chain with intensity matrix $\mathbf{Q}$). The estimation of mover-stayer model was given by Frydman \cite{Frydman1984}. The extension to a mixture of two continuous-time Markov chains moving with different speeds ($\psi_i\neq 0$) was proposed by Frydman \cite{Frydman2005}. Frydman and Schuermann \cite{Frydman2008} later on used the mixture process to model the dynamics of firms' credit ratings. As empirically shown in \cite{Frydman2008}, there is strong evidence that firms of the same credit rating may move at different speed to other credit ratings, a feature that lacks in the Markov model.

The mixture is defined conditionally on the initial state. We denote by $\phi=\mathbf{1}_{\{X=X^G\}}$ random variable having value $1$ when $X=X^G$ and $0$ otherwise, i.e.,
\begin{equation}\label{eq:mixture}
\begin{split}
&X=
\begin{cases}
X^G, &\phi =1\\
X^Q, &\phi =0.
\end{cases}
\end{split}
\end{equation}
Notice that when the random variable $\phi$ is independent of $X$ having Bernoulli distribution with probability of success $s_{i_0}$ (\ref{eq:portion}), (\ref{eq:mixture}) reduced to regime switching model for Markov jump process with changing intensity matrix $\mathbf{Q}$ and $\mathbf{G}$.

For initial state $X_0=i_0\in E$, there is a separate mixing distribution:
\begin{equation}\label{eq:portion}
s_{i_0}=\mathbb{P}\{\phi=1 \vert X_0=i_0\} \quad \mathrm{and} \quad 1-s_{i_0}=\mathbb{P}\{\phi=0 \vert X_0=i_0\},
\end{equation}
with $0\leq s_{i_0} \leq 1$. The quantity $s_{i_0}$ has the interpretation as being the proportion of population with initial state $i_0$ moving according to $X^G$, whilst $1-s_{i_0}$ is the proportion that propagates according to $X^Q$. It follows from (\ref{eq:GH}) that in general $X^Q$ and $X^G$ are different in rates at which the process leaves the states, i.e., $q_{i}\neq g_{i}$, but under (\ref{eq:Lambda}) both of them may have the same transition probability of leaving from state $i\in E$ to state $j\in\mathbb{S}$, $i\neq j$, i.e., $q_{ij}/q_{i}=g_{ij}/g_{i}$. Note that we have used $q_i$ and $g_i$ to denote diagonal element of $\mathbf{Q}$ and $\mathbf{G}$, respectively. Thus, depending on the value of $\psi_i$ (\ref{eq:Lambda}), $X^G$ never moves out of state $i$ when $\psi_i=0$ (reduced to the mover-stayer model), moves out of state $i$ at lower rate when $0<\psi_i<1$ and moves at higher rate when $\psi_i>1$ (or the same for $\psi_i=1$) than $X^Q$. If $\psi_i=1$, for all $i=1,...,m$, it simplifies to simple Markov process.

The main feature of the Markov mixture process $X$ (\ref{eq:mixture}) is that unlike its component $X^Q$ and $X^G$, $X$ does not necessarily inherit the Markovian property of  $X^Q$ neither $X^G$; the conditional distribution of future state of $X$ depends on its past information. Empirical evidence of this fact was given by Frydmand and Schuermann \cite{Frydman2008}. It was shown in \cite{Frydman2008} that the incorporation of the past information help improve the out-of-sample prediction of the Nelson-Aalen estimate of the cumulative intensity rate for risky bonds. This path dependence property certainly gives favorable feature in applications. Motivated by these findings, we derive the phase-type distribution (of time until absorption $\tau$ (\ref{eq:DefTime})) of $X$ and introduced forward intensity associated with the distribution. This intensity determines the rate of occurrences of future events given available information. Also, we derive the residual lifetime of $X$ and its occupation time in any state $i\in\mathbb{S}$. Below we present some preliminaries needed to establish the results.

\section{Preliminaries}

Recall that the mixture model (\ref{eq:mixture}) contains two regimes (processes), one at slow speed, the other fast. The regime is however not directly observable, but can be identified based on past information (\cite{Frydman2005}, \cite{Frydman2008}). To incorporate past information, we denote by $\mathcal{I}_{t-}=\{X(s), 0\leq s \leq t-\}$ the available realization of $X$ prior to time $t$, and by $\mathcal{I}_{i,t}=\mathcal{I}_{t-}\cup\{X_t=i\}$ all previous and current information of $X$. The set $\mathcal{I}_{t-}$ may contain complete, partial or maybe no information about $X$.

The likelihoods of observing the past realization $\mathcal{I}_{i,t}$ of $X$ under $X^Q$ and $X^G$ conditional on the initial state $X_0=i_0$ are given respectively, following \cite{Frydman2008}, by
\begin{equation}\label{eq:likelihood}
\begin{split}
L^Q:=&\mathbb{P}\{\mathcal{I}_{i,t} \vert \phi=0, X_0=i_0\}= \prod_{k\in E} e^{-q_{k} T_k} \prod_{j\neq k, j\in \mathbb{S}} (q_{kj})^{N_{kj}},\\
L^G:=&\mathbb{P}\{\mathcal{I}_{i,t} \vert \phi=1, X_0=i_0\}= \prod_{k\in E} e^{-g_{k} T_k} \prod_{j\neq k, j\in\mathbb{S}} (g_{kj})^{N_{kj}},
\end{split}
\end{equation}
where in the both expressions we have denoted subsequently by $T_k$ and $N_{kj}$ the total time the process spent in state $k$ for $\mathcal{I}_{i,t}$, $k\in \mathbb{S}$, and the number of transitions from state $k$ to state $j$, with $j\neq k$, observed in $\mathcal{I}_{i,t}$; whereas $q_{ij}$ and $g_{ij}$ represent the $(i,j)-$entry of intensity matrix $\mathbf{Q}$ and $\mathbf{G}$, respectively.

The likelihood is used to classify which regime the process $X$ belongs to. Frydman \cite{Frydman2005} proposed EM algorithm for the maximum likelihood estimation ($\widehat{s}_i,\widehat{\psi}_i, \widehat{q}_i, 1\leq i \leq m$) of the mixture model parameters. The results are given by
\begin{equation}\label{eq:estimationQ}
\widehat{q}_{ij}=\frac{N_{ij}}{\sum_{j\neq i} N_{ij}} \widehat{q}_i \quad \mathrm{and} \quad
\widehat{g}_{ij}= \widehat{\psi}_i \widehat{q}_{ij}, \quad j\neq i \in \mathbb{S}.
\end{equation}

Throughout the remaining of this paper, we define quantity $s_i(t)$ given by
\begin{equation}\label{eq:gammait}
s_i(t)=\mathbb{P}\{\phi=1\vert \mathcal{I}_{i,t}\}, \quad i\in\mathbb{S},
\end{equation}
Note that in the case where $t=0$, for which $\mathcal{I}_{i,t}=\{X_0=i_0\}$, $s_{i}(0)=s_{i_0}$ (\ref{eq:portion}).
This quantity says given that the process ends in state $i$ at time $t$ and all its prior realizations $I_{t-}$, the proportion of those moving according to $X^G$ is by $s_{i}(t)$.


Next, we denote by $\mathbf{S}_n(t)$ diagonal matrix, with $n\in\{m,m+1\}$, defined by
\begin{equation}\label{eq:Snt}
\mathbf{S}_n(t) = \mathrm{diag}(s_1(t), s_2(t),...,s_n(t)), \quad \textrm{with} \quad \mathbf{S}_n:= \mathbf{S}_n(0).
\end{equation}
Note that for the absorbing state $i=m+1$, $s_{m+1}(t)$ denotes the portion of those arriving in the absorbing state according to $X^G$. In special case when $t=0$, for which $\mathcal{I}_{i,t}=\{X_0=i\}$, $s_i$ is equal to the one given in (\ref{eq:portion}). The transition probability matrices $\mathbf{P}^Q(t)$ and $\mathbf{P}^G(t)$ of the underlying Markov chain $X^Q$ and $X^G$ are given following (\ref{eq:transprob}) by $\mathbf{P}^Q(t)=e^{\mathbf{Q}t}$ and $\mathbf{P}^G(t)=e^{\mathbf{G}t}$. The result below gives the transition matrix $\mathbf{P}$ of the mixture process in terms of $\mathbf{S}_{m+1}$, $\mathbf{P}^Q(t)$ and $\mathbf{P}^G(t)$. The result was stated in \cite{Frydman2005}. Below we give the proof of the result.

\begin{prop}\label{prop:PXt}
For a given $t\geq 0$, the transition probability matrix $\mathbf{P}(t)$ over period of time $(0,t)$ of the Markov mixture process $X$ (\ref{eq:mixture}) is given by
\begin{equation}\label{eq:PXt}
\mathbf{P}(t)=\mathbf{S}_{m+1}e^{\mathbf{G}t}+ \big(\mathbf{I}_{m+1}-\mathbf{S}_{m+1}\big)e^{\mathbf{Q}t},
\end{equation}
where we have defined by $\mathbf{I}_{n}$, $n\times n$ identity matrix with $n=m+1$.
\end{prop}
\begin{proof}
The proof follows from applying the law of total probability along with the use of Bayes' theorem for conditional probability and (\ref{eq:portion}). For $i,j\in\mathbb{S}$,
\begin{equation*}
\begin{split}
p_{ij}(t)=&\mathbb{P}\{X_t=j \vert X_0=i\}\\
=& \mathbb{P}\{X_t=j, \phi=1 \vert X_0=i\} + \mathbb{P}\{X_t=j, \phi=0 \vert X_0=i\}\\
=&  \mathbb{P}\{X_t=j \vert \phi=1, X_0=i\} \mathbb{P}\{\phi=1 \vert X_0=i\}  \\
 & +  \mathbb{P}\{X_t=j \vert \phi=0, X_0=i\} \mathbb{P}\{\phi=0 \vert X_0=i\} \\
=&  p_{ij}^G(t) s_i + p_{ij}^Q(t) (1-s_i),
\end{split}
\end{equation*}
where the last equality is due to the Markov property of the Markov chains $X^G$ and $X^Q.$
This is the $(i,j)$ entry of the transition probability matrix $\mathbf{P}$ (\ref{eq:PXt}). \exit\\
\end{proof}

Depending on the availability of the past information $\mathcal{I}_{i,t}$ of $X$, the elements $s_i(t)$, of the information matrix $\mathbf{S}_m(t)$ (\ref{eq:Snt}) are given following (\ref{eq:gammait}) below.
\begin{lem}\label{lem:lem1}
For a given $t\geq 0$ and any state $i, i_0\in \mathbb{S}$, we have
\begin{itemize}
\item[(i)] in case of full information, i.e., $\mathcal{I}_{i,t}=\{X(s), 0\leq s \leq t\}$,
\begin{equation}\label{eq:likelihood1}
s_i(t)=\frac{s_{i_0} L^G}{ s_{i_0} L^G + (1-s_{i_0}) L^Q},
\end{equation}
where $s_{i_0}$, $L^G$ and $L^Q$ are defined respectively in (\ref{eq:portion}) and (\ref{eq:likelihood}).

\item[(ii)] in case of limited information with $\mathcal{I}_{i,t}=\{X_0=i_0\}\cup\{X(t)=i\}$,
\begin{equation}\label{eq:likelihood3}
s_i(t)=\frac{ \mathbf{e}_{i_0}^{\top}\mathbf{S}_{m+1} \mathbf{P}^G(t)\mathbf{e}_i}{\mathbf{e}_{i_0}^{\top} \mathbf{P}(t) \mathbf{e}_i}.
\end{equation}

\item[(iii)] in case of limited information with $\mathcal{I}_{i,t}=\{X(t)=i\}$,
\begin{equation}\label{eq:likelihood2}
s_i(t)=\frac{\boldsymbol{\pi}^{\top}\mathbf{S}_{m+1} \mathbf{P}^G(t)\mathbf{e}_i}{\boldsymbol{\pi}^{\top} \mathbf{P}(t)\mathbf{e}_i},
\end{equation}
%
where $\mathbf{e}_i=(0,..,1,...0)^{\top}$ is $(m+1)\times1-$vector with $1$ at $i$th element.
\end{itemize}
\end{lem}
\begin{proof}
The proof of (\ref{eq:likelihood1}) is given on p.1074 in \cite{Frydman2008}, whilst (\ref{eq:likelihood2}) and (\ref{eq:likelihood3}) are the matrix notation of the corresponding quantity $s_i(t)$ given on p.1074 of \cite{Frydman2008}. \exit\\
\end{proof}

Notice that the past information $\mathcal{I}_{i,t}$ of the mixture process $X$ (\ref{eq:mixture}) is stored in matrix $\mathbf{S}_{m+1}(t)$ (\ref{eq:Snt}). Below we give the limiting value of $s_i(t)$ (\ref{eq:likelihood2}) as $t\rightarrow \infty$.
\begin{prop}\label{prop:limit}
Let $\mathbf{T}$ ($\boldsymbol{\Psi}\mathbf{T}$) have distinct eigenvalues $\varphi_j^{(1)}$ ($\varphi_j^{(2)}$), $j=1,...,m,$ with $\varphi_{p_k}^{(k)}=\max\{\varphi_j^{(k)}, j=1,...,m\}$ and $p_k=\textrm{argmax}_j\{\varphi_j^{(k)}\}$, $k=1,2$. Then,
\begin{equation}\label{eq:limit}
  \lim_{t\rightarrow \infty} s_i(t)=
  \begin{cases}
    \begin{cases}
     1, & \text{if  $\varphi_{p_2}^{(2)}>\varphi_{p_1}^{(1)}$} \\
     0, & \text{if  $\varphi_{p_2}^{(2)}<\varphi_{p_1}^{(1)}$} \\
     \frac{\boldsymbol{\pi}^{\top}\mathbf{S}_m L_{p_2}(\boldsymbol{\Psi}\mathbf{T}) \mathbf{e}_i}
     {\boldsymbol{\pi}^{\top} \big( \mathbf{S}_m L_{p_2}(\boldsymbol{\Psi}\mathbf{T})+ (\mathbf{I}_m-\mathbf{S}_m) L_{p_1}(\mathbf{T}) \big)\mathbf{e}_i},
        & \text{if  $\varphi_{p_2}^{(2)}=\varphi_{p_1}^{(1)}$}
    \end{cases}
    &\text{for $i\in E$}\\
    \begin{cases}
      \frac{\boldsymbol{\pi}^{\top}\mathbf{S}_m \mathbf{T}^{-1}\boldsymbol{\delta}}
      {\boldsymbol{\pi}^{\top}\mathbf{T}^{-1}\boldsymbol{\delta}},
      & \quad \hspace{3cm} \text{for $i=m+1$} .
    \end{cases}
  \end{cases}
\end{equation}
%
where $L_{p_1}(\mathbf{T})$ is the Lagrange interpolation coefficients defined by
\begin{equation}\label{eq:apostol2}
L_{p_1}(\mathbf{T})=\prod_{j=1,j\neq p_1}^{m} \Big(\frac{\mathbf{T}-\varphi_j^{(1)}\mathbf{I}_{m}}{\varphi_{p_1}^{(1)}-\varphi_j^{(1)}}\Big).
\end{equation}
\end{prop}

\begin{proof}
Recall that as the intensity matrices $\mathbf{T}$ and $\boldsymbol{\Psi}\mathbf{T}$ are negative definite, it is a known fact that their eigenvalues have negative real part, and that among them there is a dominant eigenvalue (with the smallest absolute value of the real part) which is unique and real (see O'cinneide \cite{Ocinneide} p.27). Then, following Theorem 2 in Apostol \cite{Apostol}, we have for $m\times m$ intensity matrix $\mathbf{T}$ ($\boldsymbol{\Psi}\mathbf{T}$), with $m$ distinct eigenvalues $\varphi_1^{(j)},...,\varphi_m^{(j)}$, for $j=1,2$, that $\exp\big(\mathbf{T}t\big)$ has an explicit form:
\begin{equation}\label{eq:apostol}
\exp\big(\mathbf{T} t\big) =\sum_{k=1}^m \exp\big(\varphi_k^{(1)} t\big) L_k(\mathbf{T}),
\end{equation}
similarly defined for $\exp\big(\boldsymbol{\Psi}\mathbf{T} t\big)$. Applying (\ref{eq:apostol}) to (\ref{eq:likelihood2}) we arrive at (\ref{eq:limit}). \exit\\
\end{proof}

The transition probability of $X$ conditional on $\mathcal{I}_{i,t}$ constitutes the basic building block of constructing conditional phase-type distributions and forward intensity for mixture models (\ref{eq:mixture}). Using Lemma \ref{lem:lem1}, this conditional probability is given below as an adaptation in matrix notation of eqn. (11) on p.1074 of \cite{Frydman2008}.
\begin{theo}\label{theo:theo1}
For any two times $s\geq t\geq 0$ and $i,j\in \mathbb{S}$,
\begin{equation}\label{eq:transM}
\mathbb{P}\{X_s=j \vert \mathcal{I}_{i,t}\}= \mathbf{e}_i^{\top} \Big(\mathbf{S}_{m+1}(t)e^{\mathbf{G}(s-t)}
+ \big[\mathbf{I}_{m+1}-\mathbf{S}_{m+1}(t)\big]e^{\mathbf{Q}(s-t)} \Big)\mathbf{e}_j,
\end{equation}
where $\mathbf{e}_i=(0,..,1,...0)^{\top}$ is $(m+1)\times1-$vector with $1$ at the $i$th element.
\end{theo}
\begin{proof}
Following similar arguments to the proof of (\ref{eq:PXt}), the result is established by applying the law of total probability along with Bayes' theorem, i.e.,
\begin{equation}\label{eq:proof}
\begin{split}
\mathbb{P}\{X_s=j \vert \mathcal{I}_{i,t}\}=&\mathbb{P}\{X_s=j, \phi=1 \vert \mathcal{I}_{i,t}\}
+ \mathbb{P}\{X_s=j, \phi=0 \vert \mathcal{I}_{i,t}\}\\
=& \mathbb{P}\{X_s=j \vert \phi=1, \mathcal{I}_{i,t}\}\mathbb{P}\{\phi=1 \vert \mathcal{I}_{i,t}\} \\
&+ \mathbb{P}\{X_s=j \vert \phi=0, \mathcal{I}_{i,t}\}\mathbb{P}\{\phi=0 \vert \mathcal{I}_{i,t}\}\\
=& p_{ij}^G(s-t) s_i(t) + p_{ij}^Q(s-t) (1-s_i(t)),
\end{split}
\end{equation}
where the last equality is due to Markov property of $X^G$ and $X^Q.$
As $(i,j)\in \mathbb{S}$, the last expression can be represented in the matrix notation by (\ref{eq:transM}). \exit\\
\end{proof}

It is clearly to see following (\ref{eq:transM}) that, unless when $\boldsymbol{\Psi}=\mathbf{I}$, $X$ does not inherit the Markov and stationary property of $X^G$ and $X^Q$, i.e., future occurrence of $X$ is determined by its past information $\mathcal{I}_{i,t}$ through its likelihoods and the age ($t$).

Putting the conditional probability (\ref{eq:transM}) into a matrix $\mathbf{P}(t,s)$ gives a $(t,s)-$transition probability matrix of the mixture process $X$ defined by
\begin{equation}\label{eq:probmat}
\mathbf{P}(t,s)=\mathbf{S}_{m+1}(t)e^{\mathbf{G}(s-t)}
+ \big[\mathbf{I}_{m+1}-\mathbf{S}_{m+1}(t)\big]e^{\mathbf{Q}(s-t)}.
\end{equation}

\begin{lem}\label{lem:composition}
The transition probability matrix (\ref{eq:probmat}) has the composition
\begin{equation}\label{eq:blockdiag1}
\begin{split}
\mathbf{P}(t,s)
=
 \left(\begin{array}{cc}
  \mathbf{F}_{11}(t,s)
  & \mathbf{F}_{12}(t,s) \\
  \mathbf{0} & 1 \\
  \end{array}\right),
\end{split}
\end{equation}
where the matrix entries $\mathbf{F}_{11}(t,s)$ and $\mathbf{F}_{12}(t,s)$ are defined by
\begin{eqnarray*}
\mathbf{F}_{11}(t,s)&=&\mathbf{S}_m(t)e^{\boldsymbol{\Psi}\mathbf{T}(s-t)} + (\mathbf{I}_m-\mathbf{S}_m(t))e^{\mathbf{T}(s-t)}, \\
\mathbf{F}_{12}(t,s)&=& \mathbf{S}_m(t)(\boldsymbol{\Psi}\mathbf{T})^{-1}\big(e^{\boldsymbol{\Psi}\mathbf{T}(s-t)}-\mathbf{I}_m\big)\boldsymbol{\Psi}\boldsymbol{\delta}\\
&&+(\mathbf{I}_m-\mathbf{S}_m(t))\mathbf{T}^{-1}\big(e^{\mathbf{T}(s-t)}-\mathbf{I}_m\big)\boldsymbol{\delta}.
\end{eqnarray*}
\end{lem}
\begin{proof}
Using the property of matrix exponential (\ref{eq:expm}), it can be shown that
\begin{equation*}\label{eq:partQ}
e^{\mathbf{Q}t} = \left(\begin{array}{cc}
  e^{\mathbf{T}t}  & \mathbf{T}^{-1}(e^{\mathbf{T}t}-\mathbf{I}_m)\boldsymbol{\delta}\\
  \mathbf{0} & 1 \\
\end{array}\right) \quad \mathrm{and} \quad
e^{\mathbf{G}t} = \left(\begin{array}{cc}
  e^{\boldsymbol{\Psi}\mathbf{T}t}  & (\boldsymbol{\Psi}\mathbf{T})^{-1}(e^{\boldsymbol{\Psi}\mathbf{T}t}-\mathbf{I}_m)\boldsymbol{\Psi}\boldsymbol{\delta}\\
  \mathbf{0} & 1 \\
\end{array}\right).
\end{equation*}
By partitioning identity matrix $\mathbf{I}_{m+1}$ and $\mathbf{S}_{m+1}$ (\ref{eq:Snt}) into block diagonal matrix
\begin{equation*}
\mathbf{I}_{m+1} = \left(\begin{array}{cc}
  \mathbf{I}_m & \mathbf{0} \\
  \mathbf{0} & 1 \\
\end{array}\right) \quad \mathrm{and} \quad
\mathbf{S}_{m+1}(t) = \left(\begin{array}{cc}
  \mathbf{S}_m(t) & \mathbf{0} \\
  \mathbf{0} & s_{m+1}(t)\\
\end{array}\right),
\end{equation*}
the transition matrix $\mathbf{P}(t,s)$ has the composition given by (\ref{eq:blockdiag1}). \exit
\end{proof}

The result gives a generalization to the transition matrix $\mathbf{P}(t)$ (\ref{eq:PXt}).

\begin{cor}
Let $\boldsymbol{\Psi}=\mathbf{I}$. Then, the transition probability of $X$ is given by
\begin{equation*}
\mathbf{P}(t,s)= e^{\mathbf{Q}(s-t)},
\end{equation*}
implying that the mixture process is a simple homogeneous Markov chain.
\end{cor}

\section{Generalized phase-type distributions}\label{section:Main}

This section presents the main results of this paper. First, we derive conditional phase-type distributions $\overline{F}_i(t,s)=\mathbb{P}\{\tau > s \vert \mathcal{I}_{i,t}\} $ and its conditional density $f_i(t,s):= -\frac{\partial}{\partial s} F_i(t,s)$ of the Markov mixtures process (\ref{eq:mixture}) given the past information $\mathcal{I}_{i,t}$. Secondly, we derive the unconditional distribution $\overline{F}(t)=\mathbb{P}\{\tau>t\}$ and its density $f(t):=-\frac{d}{dt} \overline{F}(t)$ of the phase-type and study the dense and closure properties of the distributions using its Laplace transform. Afterwards, we discuss forward intensity and express the distributions in terms of the intensity.

\subsection{\textbf{Conditional distributions}}
\begin{theo}\label{theo:mainPH}
For a given $t\geq 0$ and information set $\mathcal{I}_{i,t}$, the generalized conditional phase-type distribution and its density are given for any $s \geq t$ by
\begin{equation}\label{eq:main}
\begin{split}
\overline{F}_i(t,s) \;=&\; \mathbf{e}_i^{\top}\Big(\mathbf{S}_m(t) e^{\boldsymbol{\Psi}\mathbf{T}(s-t)}
+ \big(\mathbf{I}_m-\mathbf{S}_m(t)\big) e^{\mathbf{T}(s-t)}\Big)\mathbf{1}_m,\\
f_i(t,s)\;=& \;\mathbf{e}_i^{\top}\Big(\mathbf{S}_m(t)e^{\boldsymbol{\Psi}\mathbf{T}(s-t)}\boldsymbol{\Psi}
+ \big(\mathbf{I}_m-\mathbf{S}_m(t)\big)e^{\mathbf{T}(s-t)}\Big)\boldsymbol{\delta},
\end{split}
\end{equation}
for any initial state $i\in E$ $X$ started in at time $t$ and is zero for $i=m+1$.
\end{theo}

Under full information, we observe that the distribution has path dependence on its past information $\mathcal{I}_{i,t}$ through the appearance of matrix $\mathbf{S}_m(t)$ (\ref{eq:Snt}). However, when the information is limited to only the current (and initial) state, the distribution forms a non stationary function of time $t$. In any case, the proposed distribution has the ability to capture heterogeneity of object of interest represented by the speed matrix $\boldsymbol{\Psi}$ (\ref{eq:GH}). These two important properties are removed when we set $\boldsymbol{\Psi}$ being equal to $m\times m-$identity matrix $\mathbf{I}_m$. As a result, we have,

\begin{cor}
Let $\boldsymbol{\Psi}=\mathbf{I}_m$. Then the conditional distributions (\ref{eq:main}) reduced to
\begin{equation}\label{eq:uncondph}
\begin{split}
\overline{F}_i(t,s)\;=\; \mathbf{e}_i^{\top}e^{\mathbf{T}(s-t)}\mathbf{1}_m  \quad \mathrm{and} \quad
f_i(t,s)\;=\; \mathbf{e}_i^{\top}e^{\mathbf{T}(s-t)}\boldsymbol{\delta}.
\end{split}
\end{equation}
\end{cor}

\begin{Rem}
Observe that the conditional distribution (\ref{eq:uncondph}) can be obtained straightforwardly from (\ref{eq:DistrDefTime}) by setting $\boldsymbol{\pi}=\mathbf{e}_i$ and replacing the time $t$ in (\ref{eq:DistrDefTime}) by $s-t$. However, this is in general not possible for obtaining (\ref{eq:main}) from (\ref{eq:main2}).
\end{Rem}

Furthermore, the new distribution has two additional parameters represented by $\mathbf{S}_m$ and $\boldsymbol{\Psi}$ responsible for the inclusion of past information and modeling heterogeneity. It is fully characterized by the parameters $(\boldsymbol{\pi},\mathbf{T},\boldsymbol{\Psi},\mathbf{S}_m)$, which in turn generalizing the existing distribution (\ref{eq:DistrDefTime}) which is specified by $(\boldsymbol{\pi},\mathbf{T})$.

\medskip

\begin{proof}
As $\tau$ (\ref{eq:DefTime}) is the time until absorption of $X$, we have by (\ref{eq:transM})-(\ref{eq:probmat})
\begin{equation*}
F_i(t,s)=\sum_{j=1}^{m} \mathbb{P}\{ X_s=j \vert \mathcal{I}_{i,t}\} = \sum_{j=1}^{m} \mathbf{e}_i^{\top} \mathbf{P}(t,s)\mathbf{e}_j
=\mathbf{e}_i^{\top} \mathbf{P}(t,s)\big(\mathbf{1}_{m+1}-\mathbf{e}_{m+1}\big). \notag
\end{equation*}
On account of the fact that $\mathbf{1}_{m+1}-\mathbf{e}_{m+1}=(\mathbf{1}_m,0)^{\top}$, the claim for distribution $\overline{F}_i(t,s)$ follows on account of (\ref{eq:blockdiag1}). The density $f_i(t,s)$ is obtained by taking partial derivative of $\overline{F}_i(t,s)$ w.r.t $s$ and then applying the condition (\ref{eq:Tdone}). \exit\\
\end{proof}
%
%

Observe following the above proof that the $(m+1)$th element $s_{m+1}(t)$ of $\mathbf{S}_{m+1}(t)$ (\ref{eq:Snt}) does not play any role in getting the final result (\ref{eq:main}). This element corresponds to the portion of population arriving in the absorbing state at speed $\boldsymbol{\Psi}\boldsymbol{\delta}$. So, necessarily we can set $s_{m+1}=0$ which we will use in numerical example.

\subsection{\textbf{Unconditional distributions}}\label{sec:uncond}

Using the main result (\ref{eq:main}) for $t=0$ and summing up over all $i\in E$ for each given weight $\mathbb{P}\{X_0=i\}$, we obtain the unconditional phase-type distribution for the mixture process. This unconditional distribution will be used to study dense and closure properties of the distribution through its Laplace transform.

\begin{theo}
The unconditional distribution is given for any $t\geq 0$ by
\begin{equation}\label{eq:main2}
\begin{split}
\overline{F}(t)=&\boldsymbol{\pi}^{\top}\Big(\mathbf{S}_m e^{\boldsymbol{\Psi}\mathbf{T}t}
+ \big(\mathbf{I}_m-\mathbf{S}_m\big) e^{\mathbf{T}t}\Big)\mathbf{1}_m,\\
f(t)=&\boldsymbol{\pi}^{\top}\Big(\mathbf{S}_m e^{\boldsymbol{\Psi}\mathbf{T}t}\boldsymbol{\Psi}
+ \big(\mathbf{I}_m-\mathbf{S}_m\big) e^{\mathbf{T}t}\Big)\boldsymbol{\delta}.
\end{split}
\end{equation}
\end{theo}

For further reference throughout the remaining of this paper, we refer to $\mathrm{GPH}(\boldsymbol{\pi},\mathbf{T},\mathbf{\Psi}, \mathbf{S}_m)$ as the generalized phase-type distribution (\ref{eq:main2}) with parameters $(\boldsymbol{\pi},\mathbf{T},\mathbf{\Psi}, \mathbf{S}_m)$, and to $\mathrm{PH}(\boldsymbol{\pi},\mathbf{T})$ for the existing distribution (\ref{eq:DistrDefTime}). Setting $\mathbf{S}_m=\alpha \mathbf{I}_m$, with $0\leq \alpha \leq 1$, in (\ref{eq:main2}) leads to the convex mixture of $\mathrm{PH}(\boldsymbol{\pi},\boldsymbol{\Psi}\mathbf{T})$ and $\mathrm{PH}(\boldsymbol{\pi},\mathbf{T})$, i.e., $\mathrm{GPH}(\boldsymbol{\pi},\mathbf{T},\boldsymbol{\Psi},\alpha\mathbf{I})=\alpha \mathrm{PH}(\boldsymbol{\pi},\boldsymbol{\Psi}\mathbf{T})
+ (1-\alpha) \mathrm{PH}(\boldsymbol{\pi},\mathbf{T}).$ Thus, (\ref{eq:main2}) can be seen as a generalized mixture of the phase-type distributions (\ref{eq:DistrDefTime}).

\begin{cor}
By setting $\boldsymbol{\Psi}=\mathbf{I}_m$ in (\ref{eq:main2}), the distributions reduced to (\ref{eq:DistrDefTime}).
\end{cor}

\begin{proof}
The proof is based on similar arguments employed before to establish (\ref{eq:main}) and by conditioning $X$ on its initial position $X_0$. Using the result in (\ref{eq:main}),
\begin{eqnarray*}
\overline{F}(t) &=& \sum_{i=1}^{m+1} \sum_{j=1}^m \mathbb{P}\{X_0=i\}\mathbb{P}\{X_t=j \vert X_0=i\}\\
&=&(\boldsymbol{\pi} \;\; \pi_{m+1})^{\top}\mathbf{P}(0,t) \big(\mathbf{1}_{m+1}-\mathbf{e}_{m+1}\big),
\end{eqnarray*}
where we have used $(\boldsymbol{\pi},\pi_{m+1})^{\top}=\sum_{i=1}^{m+1} \mathbb{P}\{X_0=i\}\mathbf{e}_i^{\top}$, whilst $\pi_{m+1}$ is the mass of initial distribution on the absorbing state. Our claim for survival function $\overline{F}(t)$ is complete on account of (\ref{eq:blockdiag1}). On recalling (\ref{eq:Tdone}), we have $f(t)=-\overline{F}^{\prime}(t)$. \exit\\
\end{proof}

The theorem below gives the Laplace transform of (\ref{eq:main2}) and its $n$th moment.
\begin{theo}\label{theo:Laplace}
Let $\overline{F}$ be the phase-type distribution $\mathrm{GPH}(\boldsymbol{\pi}, \mathbf{T}, \boldsymbol{\Psi}, \mathbf{S}_m)$. Then,
\begin{itemize}
\item[(i)] the Laplace transform $\widehat{F}[\theta]=\int_0^{\infty} e^{-\theta u} f(u) du$ is given by
\begin{equation}\label{eq:Laplace}
\widehat{F}[\theta]=\boldsymbol{\pi}^{\top}\Big(\mathbf{S}_m\big(\theta\mathbf{I}_m-\boldsymbol{\Psi}\mathbf{T}\big)^{-1}\boldsymbol{\Psi}
+(\mathbf{I}_m-\mathbf{S}_m)\big(\theta\mathbf{I}_m-\mathbf{T}\big)^{-1}\Big)\boldsymbol{\delta}.
\end{equation}
\item[(ii)] the $n$th moment, \textrm{for $n=0,1,...$}, of $\tau$ (\ref{eq:DefTime}) is given by
\begin{equation}\label{eq:nthmoment}
\mathbb{E}\{\tau^n\}=(-1)^n n!\boldsymbol{\pi}^{\top}\Big(\mathbf{S}_m (\boldsymbol{\Psi}\mathbf{T})^{-n}+(\mathbf{I}_m-\mathbf{S}_m)\mathbf{T}^{-n}\Big)\mathbf{1}_m.
\end{equation}
\end{itemize}
\end{theo}
\begin{Rem}
As the conditional distribution (\ref{eq:main}) is a function of two time-variables $t$ and $s$, with $s\geq t$, we can reparameterize it in terms of initial time $t$ and remaining time $\nu:=s-t\geq 0$ such that, e.g., density can be rewritten as
\begin{equation*}
f_i(t,\nu)=\mathbf{e}_i^{\top}\Big(\mathbf{S}_m(t) e^{\boldsymbol{\Psi}\mathbf{T}\nu}\boldsymbol{\Psi} + (\mathbf{I}_m-\mathbf{S}_m(t))e^{\mathbf{T} \nu}\Big)\boldsymbol{\delta},
\end{equation*}
whose Laplace transform $\widehat{F}_t^{(i)}[\theta]=\int_0^{\infty} e^{-\theta \nu} f_{i}(t,\nu) d\nu$ is given below for any fixed $t\geq0$ and $i\in E$. The result is obtained in the same way as before to get (\ref{eq:Laplace}).
\begin{equation}\label{eq:Laplace2}
\widehat{F}_t^{(i)}[\theta]=\mathbf{e}_i^{\top}\Big(\mathbf{S}_m(t)(\theta\mathbf{I}_m-\boldsymbol{\Psi}\mathbf{T})^{-1}\boldsymbol{\Psi}
+(\mathbf{I}_m-\mathbf{S}_m(t))(\theta\mathbf{I}_m-\mathbf{T})^{-1}\Big)\boldsymbol{\delta}.
\end{equation}
So, there is similarity between the Laplace transform (\ref{eq:Laplace}) and (\ref{eq:Laplace2}). For (\ref{eq:Laplace2}) the initial distribution $\boldsymbol{\pi}$ has mass one on state $i$ and $\mathbf{S}_m$ has dependence on $t$.
\end{Rem}

\begin{proof}
As the sub-intensity matrix $\mathbf{T}$ and $\boldsymbol{\Psi}$ are nonsingular, then following Lemma 8.2.5 in \cite{Rolski} $(\theta\mathbf{I}_m-\mathbf{T})$ and $(\theta\mathbf{I}_m-\boldsymbol{\Psi}\mathbf{T})$ are nonsingular for each $\theta\geq 0$ whose all entries are rational functions of $\theta\geq 0$. Furthermore, for all $\theta\geq 0$,
\begin{equation*}
\int_0^{\infty} e^{-\theta u} e^{\mathbf{B}u} du = \big(\theta\mathbf{I}_m-\mathbf{B}\big)^{-1},
\end{equation*}
for any nonsingular matrix $\mathbf{B}$. The Laplace transform of (\ref{eq:main2}) is therefore
\begin{equation*}
\begin{split}
\widehat{F}[\theta]=&\int_0^{\infty} e^{-\theta u} \Big(\boldsymbol{\pi}^{\top}\Big( \mathbf{S}_m e^{\boldsymbol{\Psi}\mathbf{T} u} \boldsymbol{\Psi}
+(\mathbf{I}_m-\mathbf{S}_m)e^{\mathbf{T} u}\Big)\boldsymbol{\delta}\Big)du\\
=&\boldsymbol{\pi}^{\top}\mathbf{S}_m\Big(\int_0^{\infty} e^{-\theta u} e^{\boldsymbol{\Psi}\mathbf{T} u} du\Big)\boldsymbol{\Psi}\boldsymbol{\delta}
+\boldsymbol{\pi}^{\top}(\mathbf{I}_m-\mathbf{S}_m)\Big(\int_0^{\infty} e^{-\theta u} e^{\mathbf{T}u} du\Big)\boldsymbol{\delta}\\
=&\boldsymbol{\pi}^{\top}\mathbf{S}_m\big(\theta \mathbf{I}_m-\boldsymbol{\Psi}\mathbf{T}\big)^{-1}\boldsymbol{\Psi}\boldsymbol{\delta}
+\boldsymbol{\pi}^{\top}\big(\mathbf{I}_m-\mathbf{S}_m\big)\big(\theta\mathbf{I}_m-\mathbf{T}\big)^{-1}\boldsymbol{\delta}.
\end{split}
\end{equation*}
Following the same Lemma 8.2.5 of \cite{Rolski}, we have for all $\theta\geq0$ and $n\in\mathbb{N}$,
\begin{equation*}
\frac{d^n}{d\theta^n}\big(\theta\mathbf{I}_m-\mathbf{B}\big)^{-1}=(-1)^n n!\big(\theta\mathbf{I}_m-\mathbf{B}\big)^{-(n+1)},
\end{equation*}
for any nonsingular matrix $\mathbf{B}$. On account of this, $n$th derivative of $\widehat{F}[\theta]$ is
\begin{equation*}
\begin{split}
\frac{d^n}{d\theta^n}\widehat{F}[\theta]=&\boldsymbol{\pi}^{\top}\Big(\mathbf{S}_m\frac{d^n}{d\theta^n}\big(\theta\mathbf{I}_m-\boldsymbol{\Psi}\mathbf{T}\big)^{-1}\boldsymbol{\Psi}
+(\mathbf{I}_m-\mathbf{S}_m)\frac{d^n}{d\theta^n}\big(\theta\mathbf{I}_m-\mathbf{T}\big)^{-1}\Big)\boldsymbol{\delta}\\
=&(-1)^n n! \boldsymbol{\pi}^{\top}\Big(\mathbf{S}_m\big(\theta\mathbf{I}_m-\boldsymbol{\Psi}\mathbf{T}\big)^{-(n+1)}\boldsymbol{\Psi}
+(\mathbf{I}_m-\mathbf{S}_m)\big(\theta\mathbf{I}_m-\mathbf{T}\big)^{-(n+1)}\Big)\boldsymbol{\delta}.
\end{split}
\end{equation*}
The claim is established by setting $\theta=0$ and taking into account of (\ref{eq:Tdone}). \exit\\
\end{proof}

When $\boldsymbol{\Psi}=\mathbf{I}$, (\ref{eq:Laplace})-(\ref{eq:nthmoment}) reduced to the one given by Theorem 8.2.5 in \cite{Rolski}:
\begin{equation}\label{eq:existingLT}
\begin{split}
\widehat{F}[\theta]=&\boldsymbol{\pi}^{\top}(\theta\mathbf{I}_m-\mathbf{T})^{-1}\boldsymbol{\delta}, \\
\mathbb{E}\{\tau^n\}=&(-1)^n n! \boldsymbol{\pi}^{\top} \mathbf{T}^{-n} \mathbf{1}_m.\\
\end{split}
\end{equation}

\subsection{Relationship with the existing PH distribution}

In this section we express the relationship between unconditional distribution (\ref{eq:main2}) $\mathrm{GPH}(\boldsymbol{\pi},\mathbf{T},\mathbf{\Psi}, \mathbf{S}_m)$ and the existing distribution $\mathrm{PH}(\boldsymbol{\pi},\mathbf{T})$ (\ref{eq:DistrDefTime}). We give two main examples of (\ref{eq:main2}) which encloses Erlang distribution and its mixtures.

\begin{lem}\label{lem:equivalence}
The distribution $\mathrm{GPH}(\boldsymbol{\pi},\mathbf{T},\mathbf{\Psi}, \mathbf{S}_m)$ has $\mathrm{PH}(\widetilde{\boldsymbol{\pi}},\widetilde{\mathbf{T}})$ representation,
\begin{equation}\label{eq:equivalence}
\widetilde{\boldsymbol{\pi}}^{\top}=\big((\mathbf{S}_m\boldsymbol{\pi})^{\top}, ((\mathbf{I}_m-\mathbf{S}_m)\boldsymbol{\pi})^{\top}\big)
  \quad \mathrm{and} \quad \widetilde{\mathbf{T}}
  = \left(\begin{array}{cc}
  \boldsymbol{\Psi}\mathbf{T}  & \mathbf{0}\\
  \mathbf{0} & \mathbf{T} \\
  \end{array}\right).
\end{equation}
\end{lem}

\begin{Rem}
Based on the Laplace transform (\ref{eq:Laplace2}), the same conclusion is reached as in Lemma \ref{lem:equivalence} for the conditional distribution (\ref{eq:main}) for which case
\begin{equation}\label{eq:equivalence2}
\widetilde{\boldsymbol{\pi}}^{\top}=\big((\mathbf{S}_m(t)\mathbf{e}_i)^{\top}, ((\mathbf{I}_m-\mathbf{S}_m(t))\mathbf{e}_i)^{\top}\big)
\quad \mathrm{and} \quad \widetilde{\mathbf{T}}
  = \left(\begin{array}{cc}
  \boldsymbol{\Psi}\mathbf{T}  & \mathbf{0}\\
  \mathbf{0} & \mathbf{T} \\
  \end{array}\right).
\end{equation}
\end{Rem}

\begin{proof}
Following the Laplace transform (\ref{eq:Laplace}), define $\boldsymbol{\pi}_1=\mathbf{S}_m\boldsymbol{\pi}$, $\boldsymbol{\delta}_1=\boldsymbol{\Psi}\boldsymbol{\delta}$, $\boldsymbol{\pi}_2=(\mathbf{I}_m-\mathbf{S}_m)\boldsymbol{\pi}$ and $\boldsymbol{\delta}_2=\boldsymbol{\delta}$. Respectively, the pairs $(\boldsymbol{\pi}_1,\boldsymbol{\delta}_1)$ and $(\boldsymbol{\pi}_2,\boldsymbol{\delta}_2)$ define the initial proportion of population and their exit rate to absorbing state moving at speed $\boldsymbol{\Psi}\mathbf{T}$ and $\mathbf{T}$. Having defined these, we can rearrange (\ref{eq:Laplace}) as:
\begin{equation*}
\begin{split}
\widehat{F}[\theta]=&\boldsymbol{\pi}_1^{\top}\big(\theta\mathbf{I}_m-\boldsymbol{\Psi}\mathbf{T}\big)^{-1}\boldsymbol{\delta}_1
+\boldsymbol{\pi}_2^{\top}\big(\theta\mathbf{I}_m-\mathbf{T}\big)^{-1}\boldsymbol{\delta}_2\\
=&(\boldsymbol{\pi}_1^{\top} \;\; \boldsymbol{\pi}_2^{\top})
\left(\begin{array}{cc}
  (\theta\mathbf{I}_m-\boldsymbol{\Psi}\mathbf{T})^{-1}  & \mathbf{0}\\
  \mathbf{0} & (s\mathbf{I}_m-\mathbf{T})^{-1} \\
  \end{array}\right)
\left(\begin{array}{c}
  \boldsymbol{\delta}_1\\
  \boldsymbol{\delta}_2 \\
  \end{array}\right).
\end{split}
\end{equation*}
Notice that we have made the advantage of $\mathbf{S}_m$ and $(\mathbf{I}_m-\mathbf{S}_m)$ being symmetric matrices. On recalling that $(\theta\mathbf{I}_m-\boldsymbol{\Psi}\mathbf{T})$ and $(\theta\mathbf{I}_m-\mathbf{T})$ are square nonsingular matrices, the above block diagonal matrix can be further simplified as follows
\begin{equation*}
\left(\begin{array}{cc}
  (\theta\mathbf{I}_m-\boldsymbol{\Psi}\mathbf{T})^{-1}  & \mathbf{0}\\
  \mathbf{0} & (\theta\mathbf{I}_m-\mathbf{T})^{-1} \\
  \end{array}\right)= \left(\begin{array}{cc}
  (\theta\mathbf{I}_m-\boldsymbol{\Psi}\mathbf{T})  & \mathbf{0}\\
  \mathbf{0} & (\theta\mathbf{I}_m-\mathbf{T}) \\
  \end{array}\right)^{-1}.
\end{equation*}

To arrive at the representation of (\ref{eq:Laplace}) in the framework of Laplace transform (\ref{eq:existingLT}) of the existing phase-type distribution, define the following matrices
\begin{equation*}
\mathbf{I}_{2m}=\left(\begin{array}{cc}
  \mathbf{I}_m & \mathbf{0}\\
  \mathbf{0} & \mathbf{I}_m \\
  \end{array}\right) \quad \mathrm{and} \quad
  \widetilde{\mathbf{T}}= \left(\begin{array}{cc}
  \boldsymbol{\Psi}\mathbf{T}  & \mathbf{0}\\
  \mathbf{0} & \mathbf{T} \\
  \end{array}\right).
\end{equation*}

As a result of having defined the two matrices $\mathbf{I}_{2m}$ and $\widetilde{\mathbf{T}}$, the Laplace transform (\ref{eq:Laplace}) of the generalized phase-type distribution (\ref{eq:main2}) is finally given by
\begin{equation}\label{eq:LTEquiv}
\widehat{F}[\theta]=\widetilde{\boldsymbol{\pi}}^{\top} \big(\theta\mathbf{I}_{2m} - \widetilde{\mathbf{T}}\big)^{-1} \widetilde{\boldsymbol{\delta}},
\end{equation}
which falls in the form of (\ref{eq:existingLT}) with $\widetilde{\boldsymbol{\pi}}^{\top}=\big(\boldsymbol{\pi}_1^{\top},\boldsymbol{\pi}_2^{\top}\big)$ and $\widetilde{\boldsymbol{\delta}}=(\boldsymbol{\delta}_1,\boldsymbol{\delta}_2)^{\top}$. \exit
\end{proof}

\begin{Rem}\label{rem:mainrem}
The relationship between the distribution $\mathrm{GPH}(\boldsymbol{\pi},\mathbf{T},\mathbf{\Psi}, \mathbf{S}_m)$ and $\mathrm{PH}(\widetilde{\boldsymbol{\pi}},\widetilde{\mathbf{T}})$ is that the adjusted intensity matrix $\widetilde{\mathbf{T}}$ for distribution $\mathrm{PH}(\widetilde{\boldsymbol{\pi}},\widetilde{\mathbf{T}})$ is of block diagonal matrix of size $2m\times 2m$. That is to say that the transient state $E$ for the underlying Markov chains that form the phase-type distribution $\mathrm{PH}(\widetilde{\boldsymbol{\pi}},\widetilde{\mathbf{T}})$ is given by $E=E_1\cup E_2$, $E_1=\{1,...,m\}$ and $E_2=\{m+1,...,2m\}$, with initial distribution $\widetilde{\boldsymbol{\pi}}$. Denote by $\widetilde{\boldsymbol{\pi}}_1$ and $\widetilde{\boldsymbol{\pi}}_2$ the distributions of the Markov chains on the transient states $E_1$ and $E_2$, i.e., $\widetilde{\boldsymbol{\pi}}^{\top}=(\widetilde{\boldsymbol{\pi}}_1^{\top}, \widetilde{\boldsymbol{\pi}}_2^{\top})$. The corresponding distribution $\boldsymbol{\pi}$ of the Markov mixture process is determined following (\ref{eq:equivalence}) by $\boldsymbol{\pi}=\widetilde{\boldsymbol{\pi}}_1+\widetilde{\boldsymbol{\pi}}_2.$ As $\mathbf{S}_m$ is a diagonal matrix, each element $s_i$, $i=1,2,...,m$, of $\mathbf{S}_m$ is given by $s_i=\widetilde{\boldsymbol{\pi}}_1^{(i)}/\big(\widetilde{\boldsymbol{\pi}}_1^{(i)}+\widetilde{\boldsymbol{\pi}}_2^{(i)}\big)$, where $\widetilde{\boldsymbol{\pi}}_1^{(i)}$ and $\widetilde{\boldsymbol{\pi}}_2^{(i)}$ are the $i-$th element of $\widetilde{\boldsymbol{\pi}}_1$ and $\widetilde{\boldsymbol{\pi}}_2$, provided that $\boldsymbol{\pi}^{(i)}=\widetilde{\boldsymbol{\pi}}_1^{(i)}+\widetilde{\boldsymbol{\pi}}_2^{(i)}\neq 0$. Otherwise, $s_i$ can be arbitrarily chosen. Recall following (\ref{eq:equivalence}) that since $\widetilde{\mathbf{T}}=\mathrm{diag}(\boldsymbol{\Psi}\mathbf{T},\mathbf{T})$ and $\boldsymbol{\Psi}\mathbf{T}$ and $\mathbf{T}$ are square matrices, we have $e^{\widetilde{\mathbf{T}}t}=\mathrm{diag}\big(e^{\boldsymbol{\Psi}\mathbf{T}t},e^{\mathbf{T}t}\big).$ Therefore, following (\ref{eq:DistrDefTime}) we have $\overline{F}(t)=\widetilde{\boldsymbol{\pi}}^{\top}e^{\widetilde{\mathbf{T}}t}\mathbf{1}$ and $f(t)=\widetilde{\boldsymbol{\pi}}^{\top}e^{\widetilde{\mathbf{T}}t}\widetilde{\boldsymbol{\delta}}$, which are just the same results as in (\ref{eq:main2}).
\end{Rem}

Below we give two examples of unconditional phase-type distribution (\ref{eq:main2}) which encloses known distributions such as Erlang distribution and its mixture.
\begin{Ex}\label{Ex:Ex1}
Let $\mathrm{GPH}(\boldsymbol{\pi},\mathbf{T},\boldsymbol{\Psi},\mathbf{S}_m),$ have initial distribution $\boldsymbol{\pi}=(1,0,0)^{\top}$,
\begin{equation*}\label{eq:reg1B}
\mathbf{T} = \left(\begin{array}{ccc}
  -\beta_2 &  \beta_2   & 0 \\
  0        &  -\beta_2  & \beta_2 \\
  0        &  0         & -\beta_2 \\
  \end{array}\right),
\;
\boldsymbol{\Psi} = \left(\begin{array}{ccc}
  \beta_1/\beta_2 & 0                & 0 \\
  0               & \beta_1/\beta_2  & 0\\
  0               &  0               & \beta_1/\beta_2
\end{array}\right),
\;
\mathbf{S}_m = \left(\begin{array}{ccc}
  \alpha    & 0       & 0 \\
  0         & \alpha  & 0 \\
  0         &  0      & \alpha
\end{array}\right).
\end{equation*}

In view of Lemma \ref{lem:equivalence}, $\mathrm{GPH}(\boldsymbol{\pi},\mathbf{T},\boldsymbol{\Psi},\mathbf{S}_m)$ can be represented by $\mathrm{PH}\big(\widetilde{\boldsymbol{\pi}},\widetilde{\mathbf{T}}\big)$
\begin{equation*}
\widetilde{\boldsymbol{\pi}}= \left(\begin{array}{c}
  \alpha \\
  0      \\
  0     \\
  1-\alpha \\
  0 \\
  0\\
\end{array}\right),  \quad
\widetilde{\mathbf{T}}= \left(\begin{array}{cccccc}
  -\beta_1  & \beta_1  & 0        & 0        & 0        & 0 \\
  0         & -\beta_1 & \beta_1  & 0        & 0        & 0    \\
  0         & 0        & -\beta_1 & 0        & 0        & 0   \\
  0         & 0        & 0        & -\beta_2 & \beta_2  & 0 \\
  0         & 0        & 0        & 0        & -\beta_2 & \beta_2  \\
  0         & 0        & 0        & 0        & 0        & -\beta_2\\
\end{array}\right),
\end{equation*}
which is found to be the familiar form of mixture of Erlang distribution whose density function is given following (\ref{eq:main2}) and definition of exponential matrix by
\begin{equation*}
f(t)=\alpha f(t;m,\beta_1) + (1-\alpha) f(t;m,\beta_2).
\end{equation*}
where $f(t;m,\beta_i)$, for $i=1,2$ and $m=3$, is Erlang distribution defined by
\begin{equation}\label{eq:erlangpdf}
f(t;m,\beta_i)=\frac{1}{(m-1)!} t^{m-1}\beta_i^me^{-t\beta_i}.
\end{equation}
\end{Ex}
\begin{Ex}[Erlang distribution]\label{Ex:Ex2}
If we set $\alpha=1$ or $\beta_2=\beta_1$, i.e., $\boldsymbol{\Psi}=\mathbf{I}$, in Example \ref{Ex:Ex1}, it results in the Erlang distribution $\mathrm{Erl}(m,\beta_1)$ with $m=3$:
\begin{equation*}
\widetilde{\boldsymbol{\pi}}= \left(\begin{array}{c}
  1 \\
  0      \\
  0     \\
\end{array}\right),  \quad
\widetilde{\mathbf{T}}= \left(\begin{array}{cccccc}
  -\beta_1  & \beta_1  & 0         \\
  0         & -\beta_1 & \beta_1   \\
  0         & 0        & -\beta_1   \\
  \end{array}\right),
\end{equation*}
whose density $f(t;m,\beta_1)$ and distribution $F(t;m,\beta_1)$ are given by (\ref{eq:erlangpdf}) and
\begin{equation}\label{eq:erlangdist}
F(t;m,\beta_1)=1-\sum_{k=0}^{m-1} \frac{(\beta_1 t)^k}{k!} e^{-\beta_1 t}.
\end{equation}
\end{Ex}

\subsection{Closure and dense properties}

In this section we prove that the generalized phase-type distribution (\ref{eq:main}) (and \ref{eq:main2}) is closed under convex mixture and convolution, i.e., the sum and convex combinations of random variables with phase-type distributions again belong to phase-type distributions. These two properties emphasis additional importance of the phase-type distribution as a probabilistic tools in applications. For the existing phase-type distribution, the closure properties were established in \cite{Assaf1982}. We refer also to \cite{Rolski} and \cite{Breuer}. The proof is established based on Theorem \ref{theo:Laplace} and by taking the advantage of the already known closure properties of the existing distribution. Based on closure properties, we show by adapting similar arguments in \cite{Rolski} that the generalized distribution forms a dense class of distribution.

To establish the closure properties, the following result is required.

\begin{lem}
Let $p>0$ and $Z\sim \mathrm{GPH}(\boldsymbol{\pi},\mathbf{T},\mathbf{\Psi}, \mathbf{S}_m)$. Then,
\begin{equation}\label{eq:prodPH}
pZ\sim \mathrm{GPH}(\boldsymbol{\pi},\mathbf{T}/p,\mathbf{\Psi}, \mathbf{S}_m).
\end{equation}
\end{lem}
\begin{proof}
Let $\widetilde{Z}=pZ$. Following Lemma \ref{lem:equivalence} and Laplace transform (\ref{eq:LTEquiv}), we have
\begin{equation*}
\widehat{F}[\theta]=\mathbb{E}\big\{e^{-\theta \widetilde{Z}}\big\}=\widetilde{\boldsymbol{\pi}}^{\top}\big(\theta \mathbf{I}_{2m}-\widetilde{\mathbf{T}}/p\big)^{-1}\widetilde{\boldsymbol{\delta}}/p.
\end{equation*}
The proof follows on account $\widetilde{T}/p\mathbf{1}_{2m}+\widetilde{\boldsymbol{\delta}}/p=\mathbf{0},$ which establishes (\ref{eq:prodPH}). \exit\\
\end{proof}

Taking into account of Theorem \ref{theo:Laplace} on the Laplace transform along with Lemma \ref{lem:equivalence} and Remark \ref{rem:mainrem}, we show in the theorem below that the generalized phase-type distribution is closed under convex mixtures and finite convolutions.
\begin{theo}\label{theo:closure}
The generalized phase-type distribution $\mathrm{GPH}(\boldsymbol{\pi},\mathbf{T},\mathbf{\Psi}, \mathbf{S}_m)$ possesses closure property under finite convex mixtures and finite convolutions.
\end{theo}
\begin{proof}
To start with, let $Z_k\sim\mathrm{GPH}(\boldsymbol{\pi}^{(k)},\mathbf{T}^{(k)},\boldsymbol{\Psi}^{(k)},\mathbf{S}_{m}^{(k)})$, for $k=1,2$. We assume that the transient state space for $Z_1$ is given by $E_1=\{1,...,m\}$ and by $E_2=\{m+1,...,2m\}$ for $Z_2$. By representation $\mathrm{PH}(\widetilde{\boldsymbol{\pi}}^{(k)},\widetilde{\mathbf{T}}^{(k)})$, we have
\begin{equation*}
\widehat{F}^{(k)}[\theta]:=\mathbb{E}\{e^{-\theta Z_k}\}=(\widetilde{\pi}^{(k)})^{\top}\big(\theta \mathbf{I}_{2m}-\widetilde{T}^{(k)}\big)^{-1}\widetilde{\delta}^{(k)},
\end{equation*}
where the vectors $\widetilde{\pi}^{(k)}$ and $\widetilde{\delta}^{(k)}$ and matrix $\widetilde{T}^{(k)}$ are given by (\ref{eq:equivalence}):
\begin{equation*}
\widetilde{\boldsymbol{\pi}}^{(k)}=\big(\mathbf{S}_{m}^{(k)}\boldsymbol{\pi}^{(k)}, (\mathbf{I}_{m}-\mathbf{S}_{m}^{(k)})\boldsymbol{\pi}^{(k)}\big)^{\top}
  \quad \mathrm{and} \quad \widetilde{\mathbf{T}}^{(k)}
  = \left(\begin{array}{cc}
  \boldsymbol{\Psi}^{(k)}\mathbf{T}^{(k)}  & \mathbf{0}\\
  \mathbf{0} & \mathbf{T}^{(k)} \\
  \end{array}\right).
\end{equation*}
By Theorem 8.2.7 in \cite{Rolski}, it is known that $Z=pZ_1+(1-p)Z_2$ is again $\mathrm{PH}(\widehat{\boldsymbol{\pi}},\widehat{\mathbf{T}})$,
\begin{equation}\label{eq:convex}
\begin{split}
&\widehat{\boldsymbol{\pi}}=
\begin{cases}
p\widetilde{\boldsymbol{\pi}}^{(1)}, & \textrm{on $E_1$}\\
(1-p)\widetilde{\boldsymbol{\pi}}^{(2)}, & \textrm{on $E_2$}
\end{cases}
\quad \mathrm{and} \quad \widehat{\mathbf{T}}
  = \left(\begin{array}{cc}
  \widetilde{T}^{(1)}  & \mathbf{0}\\
  \mathbf{0} & \widetilde{T}^{(2)} \\
  \end{array}\right),
\end{split}
\end{equation}
for $0<p<1.$ We deduce from (\ref{eq:convex}) that new transition matrix $\widehat{\mathbf{T}}$ preserves the block diagonal nature of intensity matrix for the underlying Markov mixture process $X$ (\ref{eq:mixture}) moving according to $\widetilde{\mathbf{T}}^{(1)}$ and $\widetilde{\mathbf{T}}^{(2)}$ on the transient states $E_1$ and $E_2$ with initial distribution $p\widetilde{\boldsymbol{\pi}}^{(1)}$ and $(1-p)\widetilde{\boldsymbol{\pi}}^{(2)}$ on each of which. Taking $p=1/2$, we conclude that $Z=\frac{1}{2}(Z_1+Z_2)$ is again $\mathrm{PH}(\widehat{\boldsymbol{\pi}},\widehat{\mathbf{T}})$. Again, by the use of Laplace transform, it follows from (\ref{eq:prodPH}) that $2Z\sim \mathrm{PH}(\widehat{\boldsymbol{\pi}},\widehat{\mathbf{T}}/2)$ moving at half speed in transient states and to absorbing states. As a result, $Z_1+Z_2\sim \mathrm{PH}(\widehat{\boldsymbol{\pi}},\widehat{\mathbf{T}}/2)$ proving closure property of the distribution under finite convolutions.

Adapting similar arguments of \cite{Rolski}, the above implies that for $n\geq 2$ and any vector probability $\mathbf{p}=(p_1,...,p_n)^{\top}$, with $0<p_k<1$ and $\sum_{k=1}^n p_k=1$, the mixtures $\sum_{k=1}^n p_kZ_k$ belongs to $\mathrm{PH}(\widehat{\boldsymbol{\pi}},\widehat{\mathbf{T}})$ on transient state $E=\bigcup_{k=1}^n E_k$ with
\begin{equation*}
\begin{split}
&\widehat{\boldsymbol{\pi}}_i=p_k\widetilde{\boldsymbol{\pi}}^{(k)}\; \textrm{if $i\in E_k$}
\quad \mathrm{and} \quad \widehat{\mathbf{T}}=
\left(\begin{array}{cccc}
  \widetilde{T}^{(1)} & \mathbf{0} & \dots  & \mathbf{0} \\
    \mathbf{0} & \widetilde{T}^{(2)} &  \dots  & \mathbf{0} \\
    \vdots & \vdots & \ddots & \vdots \\
    \mathbf{0} & \mathbf{0} & \dots  & \widetilde{T}^{(n)}\\
  \end{array}\right),
\end{split}
\end{equation*}
with $\vert E_k \vert=m$. The conclusion over finite convolutions is reached when we set $p_k=1/2$ for all $k=1,...,n$ resulting in $Z=1/2\sum_{k=1}^n Z_k\sim \mathrm{PH}(\widehat{\boldsymbol{\pi}},\widehat{\mathbf{T}})$. Again, by (\ref{eq:prodPH}) $2Z\sim \mathrm{PH}(\widehat{\boldsymbol{\pi}},\widehat{\mathbf{T}}/2)$ from which we deduce $\sum_{k=1}^n Z_k\sim \mathrm{PH}(\widehat{\boldsymbol{\pi}},\widehat{\mathbf{T}}/2).$ \exit
\end{proof}

\begin{theo}[Dense property]
The family of generalized phas-type distribution (\ref{eq:main}) (as well as (\ref{eq:main2})) forms a dense class of distributions on $\mathbb{R}_+$.
\end{theo}
\begin{proof}
The proof follows from adapting the approach employed in \cite{Tijms} and \cite{Rolski}. Let $F$ be any right continuous distribution on $\mathbb{R}_+$ with $F(0)=0$. For fixed $t\geq 0$, assume that $F$ is continuous at $t$. Consider the following approximation of $F$:
\begin{equation}\label{eq:conv}
F_n(t)=\sum_{j=1}^{\infty} p_j(n) F(t;j,n),
\end{equation}
where $F(t;j,n)$ is the Erlang distribution (\ref{eq:erlangdist}) with $\beta_1=n$ and $p_j(n)$ is
\begin{equation*}
p_j(n)= F(j/n)- F((j-1)/n).
\end{equation*}
It is clear that $\sum_{j=1}^{\infty} p_j(n)=1.$ Thus, (\ref{eq:conv}) can be seen as infinite convex mixtures of Erlang distribution.
On account of (\ref{eq:erlangdist}), we have following (\ref{eq:conv})
\begin{equation*}
\begin{split}
F_n(t)=&\sum_{j=1}^{\infty}\big(F(j/n)- F((j-1)/n)\big)\sum_{k=j}^{\infty} e^{-nt}\frac{(nt)^k}{k!}\\
=&\sum_{k=0}^{\infty}e^{-nt}\frac{(nt)^k}{k!} \sum_{j=1}^k  \big(F(j/n)- F((j-1)/n)\big) \\
=& \sum_{k=0}^{\infty}e^{-nt}\frac{(nt)^k}{k!} F(k/n) \\
=& \mathbb{E}\big\{F\big(X_n(t)/n\big)\big\},
\end{split}
\end{equation*}
where $X_n(t)$ is a Poisson random variable with intensity $nt$. Since $F$ is bounded function, i.e., $\sup_{t}F(t)\leq 1$, and that $F$ is continuous at $t$, for each $\epsilon>0$ we choose $\delta>0$ such that $\vert F(x)-F(t)\vert \leq \epsilon/2$ whenever $\vert x-t \vert \leq \delta$. By doing so,
\begin{equation*}
\begin{split}
\mathbb{E}\big\{\vert F(X_n(t)/n)- F(t) \vert \big\} = & \mathbb{E}\big\{\vert F(X_n(t)/n)- F(t) \vert; \vert X_n(t)/n - t \vert \leq \delta \big\}\\
& + \mathbb{E}\big\{\vert F(X_n(t)/n)- F(t) \vert; \vert X_n(t)/n - t \vert > \delta \big\}\\
\leq & \frac{\epsilon}{2} + 2 \mathbb{P}\{\vert X_n(t)/n -t\vert >\delta\}\\
\leq & \frac{\epsilon}{2} +  \frac{2t}{\delta^2 n},
\end{split}
\end{equation*}
where we have applied Chebyshev's inequality in the last inequality and the fact that $\mathbb{E}\{X_n(t)/n\}=t$. Thus, by choosing $n$ big enough so that
$ \frac{2t}{\delta^2 n}\leq \epsilon/2$, $$\vert F_n(t)- F(t) \vert \big\} \leq \epsilon,$$ i.e., $\lim_{n\rightarrow \infty} F_n(t)=F(t)$. Next, consider the truncated version of (\ref{eq:conv}): $\widetilde{F}_n(t):=\sum_{j=1}^{n^2} p_j(n) F(t;j,n)$. By the above and monotone convergence theorem, we have $\lim_{n\rightarrow \infty} \vert \widetilde{F}_n(t) - F_n(t) \vert =0$ for each $t\geq 0$. The proof is complete by recalling Theorem \ref{theo:closure} and Example \ref{Ex:Ex1} that a mixture of Erlang distributions belongs to the class of generalized phase-type distribution. The case $F(0)\neq 0$ is handled by adding up $F(0)\delta_0$ in (\ref{eq:conv}), where $\delta_0$ is Dirac measure at $t=0$. \exit\\

\end{proof}

\subsection{Forward intensity}

In survival and event history analysis (see e.g.,\cite{Aalen2008}), it has been common practice to associate an intensity function $\lambda(t)$ to the stopping time $\tau$ (\ref{eq:DefTime}). This intensity determines rate of occurrence of an event within period of time $(t,t+dt]$ given information set $\mathcal{I}_{i,t}$ available at time $t$. Quite recently, Duan et al. \cite{Duan} introduced the notion of forward intensity which to some extent plays similar roles as forward rate of interest in finance. However, as stated in \cite{Duan} the forward intensity can only be calculated with the exact knowledge of the  underlying Markov process.

Motivated by this work, we introduce for any $s\geq t$ the forward intensity:
\begin{equation}\label{eq:intensity}
\lambda_i(t,s)ds=\mathbb{P}\{ \tau \leq s+ds \big \vert \tau > s, \mathcal{I}_{i,t}\},
\end{equation}
which corresponds to the usual intensity $\lambda_i(t)$ (see \cite{Aalen2008}) when we set $s=t$. Following (\ref{eq:intensity}), the forward intensity $\lambda_i(t,s)$ determines the rate of occurrence of future event in $(s,s+ds]$ based on the information $\mathcal{I}_{i,t}$ available at time $t$.

In the theorem below we give an explicit expression for the forward intensity.
\begin{theo}\label{theo:intensity}
For a given $s\geq t\geq 0$ and information $\mathcal{I}_{i,t}$, with any $i\in E$,
\begin{equation}\label{eq:intensity2}
\lambda_i(t,s)=\frac{\mathbf{e}_i^{\top}\big(\mathbf{S}_m(t)e^{\boldsymbol{\Psi}\mathbf{T}(s-t)}\boldsymbol{\Psi}
+ (\mathbf{I}_m-\mathbf{S}_m(t))e^{\mathbf{T}(s-t)}\big)\boldsymbol{\delta}}{\mathbf{e}_i^{\top}\big(\mathbf{S}_m(t) e^{\boldsymbol{\Psi}\mathbf{T}(s-t)}
+ (\mathbf{I}_m-\mathbf{S}_m(t)) e^{\mathbf{T}(s-t)}\big)\mathbf{1}_m},
\end{equation}
based on which the phase-type distribution $\overline{F}_i(t,s)$ and $f_i(t,s)$ (\ref{eq:main}) are given by
\begin{equation}\label{eq:survival}
\begin{split}
\overline{F}_i(t,s) =&\exp\Big(-\int_t^s \lambda_{i}(t,u) du \Big), \\
f_i(t,s) =& \lambda_i(t,s)\exp\Big(-\int_t^s \lambda_{i}(t,u) du \Big).
\end{split}
\end{equation}
\end{theo}
\begin{proof}
As $\tau$ is absolutely continuous (by (\ref{eq:main})), we have following (\ref{eq:intensity}) that
\begin{equation}\label{eq:derivation}
\begin{split}
\lambda_i(t,s)ds =&\mathbb{P}\{ \tau \leq s+ds \big \vert \tau > s, \mathcal{I}_{i,t}\} \\
=&\frac{\mathbb{P}\{ s<\tau \leq s+ds \vert \mathcal{I}_{i,t}\}}{\mathbb{P}\{\tau > s \vert \mathcal{I}_{i,t}\}} \\
=&-\frac{\mathbb{P}\{\tau>s+ds \vert \mathcal{I}_{i,t}\} - \mathbb{P}\{\tau>s \vert \mathcal{I}_{i,t}\}}{\mathbb{P}\{\tau>s \vert \mathcal{I}_{i,t}\}} \\
=&-\frac{\frac{\partial}{\partial s} \mathbb{P}\{\tau>s \vert \mathcal{I}_{i,t}\}}{\mathbb{P}\{\tau>s \vert \mathcal{I}_{i,t}\}}ds \\
=& -\frac{\frac{\partial}{\partial s} \overline{F}_i(t,s)}{\overline{F}_i(t,s)}ds.
\end{split}
\end{equation}
Solving the equation on last equality with condition $\overline{F}_{i}(t,t)=1$ gives (\ref{eq:survival}).\exit\\
\end{proof}

However, when we set $s=t$ in (\ref{eq:intensity2}), we obtain intensity rate $\lambda_i(t)$ of occurrence of an event within infinitesimal interval $(t,t+dt]$. This intensity is influenced by the available past information $\mathcal{I}_{i,t}$ of $X$ and is defined by
\begin{equation*}
\lambda_i(t)dt=\mathbb{P}\{\tau \leq t+dt \vert \tau> t, \mathcal{I}_{i,t}\}.
 \end{equation*}
To distinguish it from $\lambda(t,s)$ (\ref{eq:intensity2}), we refer to $\lambda(t)$ as instantaneous intensity.
\begin{cor}
Given $i\in E$ and past information $\mathcal{I}_{i,t}$, we have for any $t\geq 0$,
\begin{equation}\label{eq:intensity3}
\lambda_i(t)=\mathbf{e}_i^{\top}\big( \mathbf{I}_m + \mathbf{S}_m(t) (\boldsymbol{\Psi} -\mathbf{I}_m) \big)  \boldsymbol{\delta}.
\end{equation}
\end{cor}

As it plays a central role in survival analysis, it is interesting to see that the inclusion of past information of $X$, particularly when the information $\mathcal{I}_{i,t}$ is fully available, on the intensity $\lambda$ is now made possible under the mixture model.

Observe that unlike (\ref{eq:alpha}) the two intensity functions (\ref{eq:intensity2}) and (\ref{eq:intensity3}) are both dependence on past realization $\mathcal{I}_{i,t}$ of $X$. Below we give the baseline intensity function (\ref{eq:alpha}) for our generalized phase-type distributions (\ref{eq:main2}).
\begin{cor}
For given any $t>0$, the baseline function of $\tau$ is given by
\begin{equation}\label{eq:baseline}
\alpha(t)=\frac{\boldsymbol{\pi}^{\top}\big(\mathbf{S}_m e^{\boldsymbol{\Psi}\mathbf{T}t}\boldsymbol{\Psi}
+ \big(\mathbf{I}_m-\mathbf{S}_m\big) e^{\mathbf{T}t}\big)\boldsymbol{\delta}}{\boldsymbol{\pi}^{\top}\big(\mathbf{S}_m e^{\boldsymbol{\Psi}\mathbf{T}t}
+ \big(\mathbf{I}_m-\mathbf{S}_m\big) e^{\mathbf{T}t}\big)\mathbf{1}_m},
\end{equation}
based on which the distribution function $\overline{F}(t)$ and $f(t)$ (\ref{eq:main2}) are given by
\begin{equation}\label{eq:survival2}
\begin{split}
\overline{F}(t) =& \exp\Big(-\int_0^t \alpha(u) du \Big) \\
f(t) =& \alpha(t) \exp\Big(-\int_0^t \alpha(u) du \Big).
\end{split}
\end{equation}
\end{cor}
\begin{proof}
Following similar arguments for the proof of Theorem \ref{theo:intensity}, it can be shown that $\alpha(t)=f(t)/\overline{F}(t)$, where $f(t)$ and $\overline{F}(t)$ are given by (\ref{eq:main2}). \exit\\
\end{proof}

Below we give the long-run estimate of intensity (\ref{eq:intensity2}), (\ref{eq:intensity3}) and (\ref{eq:baseline}).

\begin{prop}
Let $\mathbf{T}$ ($\boldsymbol{\Psi}\mathbf{T}$) have distinct eigenvalues $\varphi_j^{(1)}$ ($\varphi_j^{(2)}$), $j=1,...,m$ with $\varphi_{p_k}^{(k)}=\max\{\varphi_1^{(k)},...,\varphi_m^{(k)}\}$ and $p_k=\textrm{argmax}_j\{\varphi_j^{(k)}\}$, $k=1,2$. Then,
\begin{equation}\label{eq:limit2}
\begin{split}
& \lim_{s\rightarrow \infty} \lambda_i(t,s)=
\begin{cases}
\frac{\mathbf{e}_i^{\top}\mathbf{S}_m(t)
L_{p_2}(\boldsymbol{\Psi}\mathbf{T})\boldsymbol{\Psi}\boldsymbol{\delta}}
{ \mathbf{e}_i^{\top}\mathbf{S}_m(t) L_{p_2}(\boldsymbol{\Psi}\mathbf{T})\mathbf{1}_m},  & \varphi_{p_2}^{(2)}>\varphi_{p_1}^{(1)} \\ \\
\frac{\mathbf{e}_i^{\top}(\mathbf{I}_m-\mathbf{S}_m(t) ) L_{p_1}(\mathbf{T})\boldsymbol{\delta}}
{ \mathbf{e}_i^{\top}( \mathbf{I}_m-\mathbf{S}_m(t) )L_{p_1}(\mathbf{T})\mathbf{1}_m}, & \varphi_{p_2}^{(2)}<\varphi_{p_1}^{(1)} \\ \\
\frac{ \mathbf{e}_i^{\top} \big( \mathbf{S}_m(t) L_{p_2}(\boldsymbol{\Psi}\mathbf{T})\boldsymbol{\Psi}
+(\mathbf{I}_m-\mathbf{S}_m(t) ) L_{p_1}(\mathbf{T}) \big) \boldsymbol{\delta} }
{  \mathbf{e}_i^{\top} \big( \mathbf{S}_m(t) L_{p_2}(\boldsymbol{\Psi}\mathbf{T})  + ( \mathbf{I}_m-\mathbf{S}_m(t) )L_{p_1}(\mathbf{T}) \big)\mathbf{1}_m   },
& \varphi_{p_2}^{(2)} = \varphi_{p_1}^{(1)}
\end{cases}
\end{split}
\end{equation}
$\forall t\geq 0$ and $i\in E$, where $L_{p_1}(\mathbf{T})$ is the Lagrange interpolation (\ref{eq:apostol2})-(\ref{eq:apostol}).
\end{prop}
\begin{proof}
Since generator matrices $\mathbf{T}$ and $\boldsymbol{\Psi}\mathbf{T}$ are both negative definite, the proof follows from applying similar arguments to the proof of Proposition \ref{prop:limit}. \exit\\
\end{proof}

\begin{cor}
Using (\ref{eq:limit}), the long-run rate of $\lambda_i(t)$ (\ref{eq:intensity3}) is given by
\begin{equation*}
\lim_{t \rightarrow \infty} \lambda_i(t)= \mathbf{e}_i^{\top}\big( \mathbf{I}_m + \mathbf{S}_m(\infty) (\boldsymbol{\Psi} -\mathbf{I}_m) \big)  \boldsymbol{\delta},
\end{equation*}
where $\mathbf{S}_m(\infty)$ is diagonal matrix whose $i$th entry is given by (\ref{eq:limit}). The long-run rate of $\alpha(t)$ (\ref{eq:baseline}) is given by inserting $t=0$ in (\ref{eq:limit2}) and replacing $\mathbf{e}_i$ by $\boldsymbol{\pi}$:
\begin{equation}\label{eq:limit3}
\begin{split}
& \lim_{t\rightarrow \infty} \alpha(t)=
\begin{cases}
\frac{ \boldsymbol{\pi}^{\top}\mathbf{S}_m L_{p_2}(\boldsymbol{\Psi}\mathbf{T})\boldsymbol{\Psi}\boldsymbol{\delta}}
{ \boldsymbol{\pi}^{\top}\mathbf{S}_m L_{p_2}(\boldsymbol{\Psi}\mathbf{T})\mathbf{1}_m}, & \varphi_{p_2}^{(2)}>\varphi_{p_1}^{(1)}\\ \\
\frac{\boldsymbol{\pi}^{\top}(\mathbf{I}_m-\mathbf{S}_m ) L_{p_1}(\mathbf{T})\boldsymbol{\delta}}
{ \boldsymbol{\pi}^{\top}( \mathbf{I}_m-\mathbf{S}_m )L_{p_1}(\mathbf{T})\mathbf{1}_m},  & \varphi_{p_2}^{(2)}<\varphi_{p_1}^{(1)} \\ \\
\frac{ \boldsymbol{\pi}^{\top} \big( \mathbf{S}_m L_{p_2}(\boldsymbol{\Psi}\mathbf{T})\boldsymbol{\Psi}
+(\mathbf{I}_m-\mathbf{S}_m  ) L_{p_1}(\mathbf{T}) \big) \boldsymbol{\delta} }
{  \boldsymbol{\pi}^{\top} \big( \mathbf{S}_m  L_{p_2}(\boldsymbol{\Psi}\mathbf{T})  + ( \mathbf{I}_m-\mathbf{S}_m  )L_{p_1}(\mathbf{T}) \big)\mathbf{1}_m   },
& \varphi_{p_2}^{(2)} = \varphi_{p_1}^{(1)}
\end{cases}
\end{split}
\end{equation}
\end{cor}

\subsection{Residual lifetime}

We derive residual lifetime $R_i(t)=\mathbb{E}\{\tau - t\vert \tau>t, \mathcal{I}_{i,t}\}$ of the mixture process (\ref{eq:mixture}) given its past information $\mathcal{I}_{i,t}$, for $i\in \mathbb{S}$. The result is summarized below.
\begin{prop}\label{prop:CTE}
For a given $t\geq0$ and past information $\mathcal{I}_{i,t}$, we have
\begin{equation}\label{eq:residue}
R_i(t)=-\mathbf{e}_i^{\top} \Big(\mathbf{S}_m(t)(\boldsymbol{\Psi}\mathbf{T})^{-1}
+(\mathbf{I}_m-\mathbf{S}_m(t))\mathbf{T}^{-1}\Big)\mathbf{1}_m.
\end{equation}
for any initial state $i\in E$ of $X$ at time $t$ and is zero for $i=m+1$.
\end{prop}

\begin{proof}
The proof is based on (\ref{eq:main}) and Bayes rule for conditional expectation,
\begin{equation}\label{eq:residuedef}
\begin{split}
\mathbb{E}\{\tau-t\vert \tau>t, \mathcal{I}_{i,t}\}=& \frac{\mathbb{E}\{(\tau-t)\mathbf{1}_{\{\tau>t\}}\vert \mathcal{I}_{i,t}\}}{\mathbb{P}\{\tau>t\vert \mathcal{I}_{i,t}\}}=\mathbb{E}\{(\tau-t)\mathbf{1}_{\{\tau>t\}}\vert \mathcal{I}_{i,t}\},
\end{split}
\end{equation}
as $\mathbb{P}\{\tau>t\vert \mathcal{I}_{i,t}\}=1$. The latter expectation can be simplified using (\ref{eq:main}) as
\begin{equation*}
\begin{split}
\mathbb{E}\{(\tau-t)\mathbf{1}_{\{\tau>t\}}\vert \mathcal{I}_{i,t}\}=& \int_t^{\infty}\mathbb{P}\{\tau>s \vert \mathcal{I}_{i,t}\} ds\\
=&\int_t^{\infty} \mathbf{e}_i^{\top} \Big(\mathbf{S}_m(t)e^{\boldsymbol{\Psi}\mathbf{T}(s-t)} + (\mathbf{I}_m-\mathbf{S}_m(t))e^{\mathbf{T}(s-t)}\Big)\mathbf{1}_mds\\
&\hspace{-2cm}=\mathbf{e}_i^{\top} \mathbf{S}_m(t) \Big( \int_t^{\infty} e^{\boldsymbol{\Psi}\mathbf{T} (s-t)} ds\Big) \mathbf{1}_m
+  \mathbf{e}_i^{\top} \big(\mathbf{I}_m-\mathbf{S}_m(t) \big) \Big(\int_t^{\infty} e^{\mathbf{T}(s-t)}ds\Big)\mathbf{1}_m.
\end{split}
\end{equation*}
It is clear following (\ref{eq:main}) that $R_i(t)=0$ for $i=m+1$ as $F_i(t,s)=0$ for $s\geq t$.
The result is established as for negative definite $\mathbf{A}$, $\int_t^{\infty} e^{\mathbf{A}(s-t)} ds= -\mathbf{A}^{-1}$. \exit\\
\end{proof}

Observe that, if $\boldsymbol{\Psi}\neq\mathbf{I}$, the expected residual lifetime has dependence on past realization of $X$ through $\mathbf{S}_m(t)$ and the current age $t$ when the information is fully available. However, when the information is limited to only the current (and initial state), the expected residual lifetime forms a non stationary function of the current age $t$. Otherwise, residual lifetime expectancy is just a constant.

\subsection{Occupation time}
The second quantity we are interested in is the generalization of the result we had above for residual lifetime expectancy. In this section, we consider the total time $T_j(t)$ the mixture process $X$ (\ref{eq:mixture}) spent in state $j\in \mathbb{S}$ up to time $t$, i.e., $$T_j(t)=\int_0^t\mathbf{1}_{\{X_s=j\}}ds.$$

This quantity somehow appears in the likelihood function (\ref{eq:likelihood}). As we allow the inclusion of past information $\mathcal{I}_{i,t}$ and the current age $t$, we define
\begin{equation}
\overline{T}_{j}(t,s)=T_j(s) - T_j(t), \quad \textrm{for $s\geq t$}.
\end{equation}
\begin{theo}\label{theo:occupation}
For a given $s\geq t\geq0$ and past information $\mathcal{I}_{i,t}$, we have
\begin{equation}\label{eq:resolvent}
\begin{split}
\mathbb{E}\{\overline{T}_{j}(t,s)\big \vert \mathcal{I}_{i,t}\}=
\begin{cases}
\mathbf{e}_i^{\top} \overline{\mathbf{F}}_{11}(t,s)\mathbf{e}_j, & \forall i,j\in E\\
\mathbf{e}_i^{\top} \overline{\mathbf{F}}_{12}(t,s)\mathbf{e}_{m+1}, & \forall i\in E \\
0, & \forall j\in E, i=m+1\\
(s-t), & \textrm{for $i,j=m+1$},
\end{cases}
\end{split}
\end{equation}
where the matrices $\overline{\mathbf{F}}_{11}(t,s)$ and $\overline{\mathbf{F}}_{12}(t,s)$ are defined by
\begin{eqnarray*}
\overline{\mathbf{F}}_{11}(t,s)&=&\mathbf{S}_m(t)\big[\boldsymbol{\Psi}\mathbf{T}\big]^{-1}
\big(e^{\boldsymbol{\Psi}\mathbf{T}(s-t)}-\mathbf{I}_m\big)\\
&&+\big(\mathbf{I}_m-\mathbf{S}_m(t)\big)\big[\mathbf{T}\big]^{-1}\big(e^{\mathbf{T}(s-t)}-\mathbf{I}_m\big),\\
\overline{\mathbf{F}}_{12}(t,s)&=& \mathbf{S}_m(t)\Big[(s-t)\mathbf{I}_m-[\boldsymbol{\Psi}\mathbf{T}]^{-1}\big(e^{\boldsymbol{\Psi}\mathbf{T}(s-t)}-\mathbf{I}_m\big)\Big]\mathbf{1}_m\\
&&+(\mathbf{I}_m-\mathbf{S}_m(t))\Big[(s-t)\mathbf{I}_m-[\mathbf{T}]^{-1}\big(e^{\mathbf{T}(s-t)}-\mathbf{I}_m\big)\Big]\mathbf{1}_m.
\end{eqnarray*}

\end{theo}
\begin{proof}
The proof follows from applying Theorem \ref{theo:theo1} to obtain
\begin{equation*}
\begin{split}
\mathbb{E}\{\overline{T}_{j}(t,s)\big \vert \mathcal{I}_{i,t}\}=& \int_t^s \mathbb{P}\{ X_u=j\vert \mathcal{I}_{i,t}\} du\\
=&\int_t^s \mathbf{e}_i^{\top} \Big(\mathbf{S}_{m+1}(t)e^{\mathbf{G}(u-t)}
+ \big[\mathbf{I}_{m+1}-\mathbf{S}_{m+1}(t)\big]e^{\mathbf{Q}(u-t)} \Big)\mathbf{e}_j du\\
=&\mathbf{e}_i^{\top} \Big(\int_t^s \left(\begin{array}{cc}
  \mathbf{F}_{11}(t,u)
  & \mathbf{F}_{12}(t,u) \\
  \mathbf{0} & 1 \\
  \end{array}\right) du \Big)\mathbf{e}_j,
\end{split}
\end{equation*}
where the last equality is due to Lemma \ref{lem:composition}, which proving the claim. \exit
\end{proof}

\section{Multi absorbing states and competing risks}\label{sec:main3}

In this section we extend the lifetime distribution of Markov mixture process to the case where we have $p\geq 1$ absorbing states and $m\geq 1$ transient states. In this situation, $\mathbf{Q}$ (\ref{eq:MatQ}) is a $(m+p)\times(m+p)$ singular matrix whose top left $(m\times m)-$submatrix $\mathbf{T}$ corresponds to $X^Q$ moving to transient states $E=\{1,...,m\}$, whereas $\boldsymbol{\delta}$ now refers to a $m\times p$ exit transition matrix of $X^Q$. To distinguish it from one absorbing state, we reserve matrix notation $\mathbf{D}$ for $\boldsymbol{\delta}$, i.e.,
\begin{equation}\label{eq:MatQ2}
\mathbf{Q} = \left(\begin{array}{cc}
  \mathbf{T} & \mathbf{D} \\
  \mathbf{0} & \mathbf{0} \\
\end{array}\right),
\end{equation}
Following (\ref{eq:matq}), both matrices $\mathbf{T}$ and $\mathbf{D}$ satisfy the following constraint:
\begin{equation}\label{eq:const2}
\mathbf{T}\mathbf{1}_{m} +\mathbf{D}\mathbf{1}_{p}=\mathbf{0}_{m},
\end{equation}
where we denoted by $\mathbf{1}_{n}$, $n=m,p$, a $(n\times 1)$unit vector of size $n$. The intensity matrix for $X^G$ is defined by to (\ref{eq:GH}). Accordingly, the information matrix $\mathbf{S}_n(t)$ (\ref{eq:Snt}) is defined as $n\times n$ diagonal matrix, $n\in\{m,m+p\}$, whose $i$th element is given by $s_i(t)$ defined in Lemma \ref{lem:lem1}, necessarily $s_i(t)=0$ for $i\in\{m+1,...,m+p\}.$

We denote by $\tau_j$ the lifetime of the Markov mixture process $X$ (\ref{eq:mixture}) until it moves to absorbing state $\delta_j:=m+j$, i.e., $\tau_j=\inf\{t>0: X_t=\delta_j\}$. In competing risks (for which we refer among others to Kalbfleisch and Prentice \cite{Kalbfleisch}) only the first event is observed, we therefore define observed lifetime $\tau=\min\{\tau_1,...,\tau_p\}$ and a random variable $\mathbf{J}$ to denote an event-specific type, i.e., $\mathbf{J}=\textrm{argmin}_j(\tau_j)$.

\subsection{Conditional distributions and intensity}\label{sec:condflamb}

We are interested in determining conditional joint distribution $F_i(t,s)$ for a given past information $\mathcal{I}_{i,t}$ of lifetime $\tau$ of the Markov mixture process (\ref{eq:mixture}) and the random variable $\mathbf{J}$ corresponding to absorption of $X$ to absorbing state $\delta_j$,
\begin{equation}\label{eq:distrfuncj}
F_{ij}(t,s)=\mathbb{P}\{\tau\leq s, \mathbf{J}=j \vert \mathcal{I}_{i,t}\} \quad \mathrm{and} \quad
\overline{F}_{ij}(t,s)=\mathbb{P}\{\tau > s, \mathbf{J}=j \vert \mathcal{I}_{i,t}\}
\end{equation}
with density $f_{ij}(t,s)=\frac{\partial}{\partial s} F_{ij}(t,s)$. By (\ref{eq:distrfuncj}) and (\ref{eq:main}), we have the following
\begin{equation}\label{eq:totalFi}
F_i(t,s)=\sum_{j=1}^p F_{ij}(t,s) \quad \mathrm{and} \quad \overline{F}_i(t,s)=1-F_i(t,s),
\end{equation}
with density $f_i(t,s)=\frac{\partial}{\partial s} F_i(t,s)$. Denote by $\lambda_{ij}(t,s)$ forward intensity associated with absorption to absorbing state $\delta_j$ for given past information $\mathcal{I}_{i,t}$ defined by
\begin{equation*}\label{eq:intensitybyj}
\lambda_{ij}(t,s)ds=\mathbb{P}\big\{\tau \leq s+ds, \mathbf{J}=j \vert \tau>s, \mathcal{I}_{i,t} \big\},
\end{equation*}
from which the overall intensity $\lambda_i(t,s)$ (\ref{eq:intensity}) for only one type of absorption is
\begin{equation}\label{eq:totintensity}
\lambda_i(t,s)=\sum_{j=1}^p \lambda_{ij}(t,s).
\end{equation}
\begin{theo}\label{theo:main3}
Given $\mathcal{I}_{i,t}$, we have for any $s \geq t\geq 0$ and $j\in\{1,...,p\}$ that
\begin{equation}\label{eq:competing}
\begin{split}
F_{ij}(t,s) =& \mathbf{e}_i^{\top}\Big(\mathbf{S}_m(t) (\boldsymbol{\Psi}\mathbf{T})^{-1}\big(e^{\boldsymbol{\Psi}\mathbf{T}(s-t)}-\mathbf{I}_m\big)\boldsymbol{\Psi}  \\
& \hspace{0.85cm}+ \big(\mathbf{I}_m-\mathbf{S}_m(t)\big) T^{-1}\big(e^{\mathbf{T}(s-t)}-\mathbf{I}_m\big)\Big)\mathbf{D} \mathbf{e}_j \\
f_{ij}(t,s) =& \mathbf{e}_i^{\top}\Big(\mathbf{S}_m(t) e^{\boldsymbol{\Psi}\mathbf{T}(s-t)}\boldsymbol{\Psi}
+ \big(\mathbf{I}_m-\mathbf{S}_m(t)\big) e^{\mathbf{T}(s-t)} \Big)\mathbf{D} \mathbf{e}_j.
\end{split}
\end{equation}
\end{theo}
\begin{proof}
As the two events $\{\tau\leq s, \mathbf{J}=j\}$ and $\{X_s=\delta_j\}$ are equivalent, we have
\begin{equation*}
\begin{split}
\mathbb{P}\big\{\tau\leq s, \mathbf{J}=j \big\vert \mathcal{I}_{i,t}\big\} =& \mathbb{P}\big\{ X_s=j \big \vert \mathcal{I}_{i,t}\big\}
=\mathbf{e}_i^{\top}\mathbf{P}(t,s) \mathbf{e}_{m+j}\\
=& \mathbf{e}_i^{\top}\Big(\mathbf{S}_m(t) (\boldsymbol{\Psi}\mathbf{T})^{-1}\big(e^{\boldsymbol{\Psi}\mathbf{T}(s-t)}-\mathbf{I}_m\big)\boldsymbol{\Psi}\\
& \hspace{0.85cm}+ \big(\mathbf{I}_m-\mathbf{S}_m(t)\big) T^{-1}\big(e^{\mathbf{T}(s-t)}-\mathbf{I}_m\big)\Big)\mathbf{D} \mathbf{e}_j.
\end{split}
\end{equation*}
Note that in the last equality we have used (\ref{eq:blockdiag1}) with $\boldsymbol{\delta}$ replaced by $\mathbf{D}$. Our claim for density $f_{ij}(t,s)$ is established on account $f_{ij}(t,s)=\frac{\partial}{\partial s} F_{ij}(t,s)$. \exit\\
\end{proof}

\begin{cor}\label{cor:cor0}
By (\ref{eq:totalFi}) and (\ref{eq:const2}), we have for any $s \geq t\geq 0$ and $i\in E$,
\begin{equation}\label{eq:rescor1}
\begin{split}
\overline{F}_{i}(t,s) =& \mathbf{e}_i^{\top}\Big(\mathbf{S}_m(t) e^{\boldsymbol{\Psi}\mathbf{T}(s-t)}
+ \big(\mathbf{I}_m-\mathbf{S}_m(t)\big) e^{\mathbf{T}(s-t)} \Big) \mathbf{1}_m\\
f_{i}(t,s) =& \mathbf{e}_i^{\top}\Big(\mathbf{S}_m(t) e^{\boldsymbol{\Psi}\mathbf{T}(s-t)}\boldsymbol{\Psi}
+ \big(\mathbf{I}_m-\mathbf{S}_m(t)\big) e^{\mathbf{T}(s-t)} \Big)\mathbf{D} \mathbf{1}_p.
\end{split}
\end{equation}
\end{cor}

It can be shown using the same steps as in (\ref{eq:derivation}) and following (\ref{eq:totalFi})-(\ref{eq:totintensity}),
\begin{equation}\label{eq:intensities}
\lambda_{ij}(t,s)=\frac{f_{ij}(t,s)}{\overline{F}_i(t,s)} \quad \mathrm{and} \quad \lambda_{i}(t,s)=\frac{f_{i}(t,s)}{\overline{F}_i(t,s)}.
\end{equation}

From the result of Theorem \ref{theo:main3}, we see that the lifetime $\tau$ and type-specific absorption are not independent. This implies following Theorem 1.1 in \cite{Crowder} that the intensity is not of proportional type as we can see from Corollary \ref{cor:cor1} below.

\begin{cor}\label{cor:cor1}
For any $i\in E$ and $j\in\{m+1,...,m+p\}$, we have $\forall s\geq t\geq 0$
\begin{equation}\label{eq:intensities2}
\begin{split}
\lambda_{ij}(t,s)=&\frac{ \mathbf{e}_i^{\top}\big(\mathbf{S}_m(t) e^{\boldsymbol{\Psi}\mathbf{T}(s-t)}\boldsymbol{\Psi}
+ \big(\mathbf{I}_m-\mathbf{S}_m(t)\big) e^{\mathbf{T}(s-t)} \big)\mathbf{D} \mathbf{e}_j   }{ \mathbf{e}_i^{\top}\big(\mathbf{S}_m(t) e^{\boldsymbol{\Psi}\mathbf{T}(s-t)}
+ \big(\mathbf{I}_m-\mathbf{S}_m(t)\big) e^{\mathbf{T}(s-t)} \big) \mathbf{1}_m    }. 
\end{split}
\end{equation}
\end{cor}

\begin{Rem}
Let $\boldsymbol{\Psi}=\mathbf{I}$, i.e., the case when the two Markov chains $X^Q$ and $X^G$ move at the same speed, the type-specific distribution $F_{ij}(t)$ (\ref{eq:competing}) and intensity $\lambda_{ij}(t)$ (\ref{eq:intensities2}) are reduced to simpler expression with information $\mathcal{I}_{i,t}$ removed:
\begin{equation*}
\begin{split}
F_{ij}(t,s)=\mathbf{e}_i^{\top} \mathbf{T}^{-1}\big( e^{\mathbf{T}(s-t)} - \mathbf{I}_m\big) \mathbf{D} \mathbf{e}_j  \quad \mathrm{and} \quad
\lambda_{ij}(t,s)= \frac{\mathbf{e}_i^{\top} e^{\mathbf{T}(s-t)} \mathbf{D} \mathbf{e}_j}{\mathbf{e}_i^{\top} e^{\mathbf{T}(s-t)} \mathbf{1}_m}.
\end{split}
\end{equation*}
\end{Rem}

In the proposition below, we give the long-term rate of intensity $\lim_{s\rightarrow \infty} \lambda_{ij}(t,s)$. The result is due to applying similar arguments used in (\ref{eq:apostol}).

\begin{prop}
Let $\mathbf{T}$ ($\boldsymbol{\Psi}\mathbf{T}$) have distinct eigenvalues $\varphi_j^{(1)}$ ($\varphi_j^{(2)}$), $j=1,...,m$ with $\varphi_{p_k}^{(k)}=\max\{\varphi_1^{(k)},...,\varphi_m^{(k)}\}$ and $p_k=\textrm{argmax}_j\{\varphi_j^{(k)}\}$, $k=1,2$. Then,
\begin{equation}\label{eq:limit5}
\begin{split}
& \lim_{s\rightarrow \infty} \lambda_{ij}(t,s)=
\begin{cases}
\frac{\mathbf{e}_i^{\top}\mathbf{S}_m(t)
L_{p_2}(\boldsymbol{\Psi}\mathbf{T})\boldsymbol{\Psi}\mathbf{D}\mathbf{e}_j}
{ \mathbf{e}_i^{\top}\mathbf{S}_m(t) L_{p_2}(\boldsymbol{\Psi}\mathbf{T})\mathbf{1}_m}, & \varphi_{p_2}^{(2)} > \varphi_{p_1}^{(1)}\\ \\
\frac{\mathbf{e}_i^{\top}(\mathbf{I}_m-\mathbf{S}_m(t) ) L_{p_1}(\mathbf{T}) \mathbf{D}\mathbf{e}_j}
{ \mathbf{e}_i^{\top}( \mathbf{I}_m-\mathbf{S}_m(t) )L_{p_1}(\mathbf{T})\mathbf{1}_m} & \varphi_{p_2}^{(2)} < \varphi_{p_1}^{(1)}  \\  \\
\frac{\mathbf{e}_i^{\top} \big( \mathbf{S}_m(t)
L_{p_2}(\boldsymbol{\Psi}\mathbf{T})\boldsymbol{\Psi}   +  (\mathbf{I}_m-\mathbf{S}_m(t) ) L_{p_1}(\mathbf{T})    \big) \mathbf{D} \mathbf{e}_j}
{ \mathbf{e}_i^{\top}  \big(  \mathbf{S}_m(t) L_{p_2}(\boldsymbol{\Psi}\mathbf{T}) + ( \mathbf{I}_m-\mathbf{S}_m(t) )L_{p_1}(\mathbf{T})  \big)  \mathbf{1}_m}, & \varphi_{p_2}^{(2)}  = \varphi_{p_1}^{(1)}
\end{cases}
\end{split}
\end{equation}
where in the two results $L_{p_1}(\mathbf{T})$ refers to the Lagrange interpolation (\ref{eq:apostol2}).
\end{prop}

Below we give the result on type-specific instantaneous intensity $\lambda_{ij}(t)$:
\begin{equation*}
\lambda_{ij}(t)dt=\mathbb{P}\{\tau \leq t+dt, \mathbf{J}=j \vert \tau>t, \mathcal{I}_{i,t}\}.
\end{equation*}
\begin{cor}\label{cor:cor2}
Given $i\in E$, $j=\{m+1,...,m+p\}$ and past information $\mathcal{I}_{i,t}$,
\begin{equation}
\lambda_{ij}(t)=\mathbf{e}_i^{\top}\big(\mathbf{I}_m + \mathbf{S}_m(t)(\boldsymbol{\Psi}-\mathbf{I}_m)\big)\mathbf{D}\mathbf{e}_j  \quad \forall t\geq 0
\end{equation}
Denote $\mathbf{S}_m(\infty)=\textrm{diag}(s_1(\infty),...,s_m(\infty))$ with $s_i(\infty)=\lim_{t\rightarrow \infty} s_i(t)$ (\ref{eq:limit}). Then,
\begin{equation*}
\lim_{t\rightarrow \infty} \lambda_{ij}(t)= \mathbf{e}_i^{\top}\big(\mathbf{I}_m + \mathbf{S}_m(\infty)(\boldsymbol{\Psi}-\mathbf{I}_m)\big)\mathbf{D}\mathbf{e}_j.
\end{equation*}
\end{cor}

Solving (\ref{eq:intensities}) for distributions $f_{i}(t,s)$ and $\overline{F}_i(t,s)$ leads to (\ref{eq:survival}) with $\lambda_i(t,s)$ given in (\ref{eq:intensities2}). By doing so, we can express $f_{ij}(t,s)$ and $F_{ij}(t,s)$ as follows.
\begin{prop}\label{eq:seeminglyindp}
Given $i\in E$ and $j=m+1,...,m+p$, we have for all $s\geq t\geq 0$,
\begin{equation}\label{eq:fFlambda}
\begin{split}
f_{ij}(t,s)=&\lambda_{ij}(t,s)\exp\Big( -\int_t^s \sum_{j=1}^p \lambda_{ij}(t,u) du\Big)\\
F_{ij}(t,s)=& \int_t^s \lambda_{ij}(t,u)\exp\Big( -\int_t^u \sum_{j=1}^p \lambda_{ij}(t,v) dv\Big) du.
\end{split}
\end{equation}
\end{prop}
\begin{Rem}
Following (\ref{eq:fFlambda}), we may consider a function $H_{ij}(t,s)=\exp\big(-\int_t^s \lambda_{ij}(t,u) du\big)$, for $s\geq t$, $j=m+1,...,m+p$ and $i=1,...,m$. However, it does not have a clear probability interpretation in view of competing risks.
\end{Rem}

Inspired by similar type of intensity discussed in Crowder \cite{Crowder}, it may seem natural to define an intensity function $\widetilde{\lambda}_{ij}(t,s)= f_{ij}(t,s)/\overline{F}_{ij}(t,s)$ given by
\begin{equation}\label{eq:otherintensity}
\begin{split}
\widetilde{\lambda}_{ij}(t,s)=&\mathbb{P}\big\{\tau \leq s+ds \big \vert \tau >s, \mathbf{J}=j, \mathcal{I}_{i,t}\big\}.
\end{split}
\end{equation}

This intensity can be seen as conditional density of lifetime distribution of Markov mixture process $X$ at absorption caused by type $j$ given its past information $\mathcal{I}_{i,t}$. Following \cite{Crowder} and by (\ref{eq:otherintensity}), it is then sensible to define distributions
\begin{equation}\label{eq:tildeF}
\widetilde{F}_{ij}(t,s)=\mathbb{P}\{\tau \leq s \vert \mathbf{J}=j,\mathcal{I}_{i,t}\}
\quad \mathrm{and} \quad \widetilde{p}_{ij}(t,s)=\mathbb{P}\{\mathbf{J}=j \vert \tau >s, \mathcal{I}_{i,t}\}.
\end{equation}

\begin{prop}
For any $i=1,...,m$, $j=m+1,...,m+p$, we have $\forall s\geq t\geq 0$,
\begin{equation*}
\begin{split}
\widetilde{F}_{ij}(t,s)=& -\Big(\frac{  \mathbf{e}_i^{\top} \mathbf{S}_m(t) (\boldsymbol{\Psi}\mathbf{T})^{-1}\big(e^{\boldsymbol{\Psi}\mathbf{T}(s-t)}-\mathbf{I}_m\big)\boldsymbol{\Psi}\mathbf{D} \mathbf{e}_j  }
{  \mathbf{e}_i^{\top} \mathbf{T}^{-1} \mathbf{D} \mathbf{e}_j  } \\ & \hspace{1.5cm}+ \frac{  \mathbf{e}_i^{\top} \big(\mathbf{I}_m-\mathbf{S}_m(t)\big) T^{-1}\big(e^{\mathbf{T}(s-t)}-\mathbf{I}_m\big)\mathbf{D} \mathbf{e}_j  }
{  \mathbf{e}_i^{\top} \mathbf{T}^{-1} \mathbf{D} \mathbf{e}_j  } \Big) \\
\widetilde{p}_{ij}(t,s)=& -\frac{\mathbf{e}_i^{\top}\big(\mathbf{S}_m(t)e^{\boldsymbol{\Psi}\mathbf{T}(s-t)} + (\mathbf{I}_m-\mathbf{S}_m(t)) e^{\mathbf{T}(s-t)} \Big) \mathbf{T}^{-1}\mathbf{D}\mathbf{e}_j   }{ \mathbf{e}_i^{\top}\Big(\mathbf{S}_m(t) e^{\boldsymbol{\Psi}\mathbf{T}(s-t)}
+ \big(\mathbf{I}_m-\mathbf{S}_m(t)\big) e^{\mathbf{T}(s-t)} \big) \mathbf{1}_m    }.
\end{split}
\end{equation*}
\end{prop}
\begin{proof}
By Bayes' rule for conditional probability, it follows from (\ref{eq:tildeF}) that
\begin{equation*}
\widetilde{F}_{ij}(t,s)=\frac{F_{ij}(t,s)}{p_{ij}} \quad  \mathrm{and} \quad \widetilde{p}_{ij}(t,s)=\frac{\overline{F}_{ij}(t,s)}{\overline{F}_i(t,s)},
\end{equation*}
where $p_{ij}$ is the probability of ultimate absorption to absorbing state $\delta_j$ given by
\begin{equation}\label{eq:ultimate}
p_{ij}=\mathbb{P}\{\mathbf{J}=j\vert \mathcal{I}_{i,t}\}=-\mathbf{e}_i^{\top} \mathbf{T}^{-1} \mathbf{D} \mathbf{e}_j
\end{equation}
The results follow from (\ref{eq:ultimate}), (\ref{eq:rescor1}), (\ref{eq:competing}) and by $\overline{F}_{ij}(t,s)=p_{ij} - F_{ij}(t,s)$:
\begin{equation}\label{eq:ultimate2}
\overline{F}_{ij}(t,s)= -\mathbf{e}_i^{\top}\big(\mathbf{S}_m(t)e^{\boldsymbol{\Psi}\mathbf{T}(s-t)} + (\mathbf{I}_m-\mathbf{S}_m(t)) e^{\mathbf{T}(s-t)} \Big) \mathbf{T}^{-1}\mathbf{D}\mathbf{e}_j. \quad \exit
\end{equation}
\end{proof}

\subsection{Unconditional distributions and intensity}
In this section we present the results on unconditional distributions
\begin{equation*}
F_{j}(t)=\mathbb{P}\{\tau \leq t, \mathbf{J}=j\} \quad \mathrm{and} \quad \overline{F}_{j}(t)=\mathbb{P}\{\tau > t, \mathbf{J}=j\},
\end{equation*}
with density $f_j(t)=F_j^{\prime}(t)$. Type-specific instantaneous intensity is defined by
\begin{equation}\label{eq:alpha2}
\alpha_j(t)dt=\mathbb{P}\{\tau\leq t+dt,, \mathbf{J}=j \vert \tau >t\}.
\end{equation}
The distribution $\overline{F}(t)$ (\ref{eq:main2}) and overall intensity $\alpha(t)$ (\ref{eq:baseline}) are given by
\begin{equation}\label{eq:sumF}
\overline{F}(t)=\sum_{j=1}^p \overline{F}_j(t) \quad \mathrm{and} \quad \alpha(t)=\sum_{j=1}^p \alpha_j(t).
\end{equation}

Based on similar arguments given in Sections \ref{sec:uncond} and \ref{sec:condflamb} we have:
\begin{theo}
Given any $j\in\{1,...,p\}$, we have for all $t\geq 0$ that
\begin{equation}\label{eq:competing2}
\begin{split}
F_{j}(t) =& \boldsymbol{\pi}^{\top}\Big(\mathbf{S}_m (\boldsymbol{\Psi}\mathbf{T})^{-1}\big(e^{\boldsymbol{\Psi}\mathbf{T} t}
-\mathbf{I}_m\big)\boldsymbol{\Psi} + \big(\mathbf{I}_m-\mathbf{S}_m\big) T^{-1}\big(e^{\mathbf{T}t}-\mathbf{I}_m\big)\Big)\mathbf{D} \mathbf{e}_j \\
&\hspace{1cm} f_{j}(t) = \boldsymbol{\pi}^{\top}\Big(\mathbf{S}_m e^{\boldsymbol{\Psi}\mathbf{T}t}\boldsymbol{\Psi}
+ \big(\mathbf{I}_m-\mathbf{S}_m \big) e^{\mathbf{T}t} \Big)\mathbf{D} \mathbf{e}_j.
\end{split}
\end{equation}
\end{theo}
Applying the same arguments of (\ref{eq:ultimate})-(\ref{eq:ultimate2}) to the result above, we obtain
\begin{equation}\label{eq:relation3}
\overline{F}_{j}(t)=-\boldsymbol{\pi}^{\top}\Big(\mathbf{S}_m e^{\boldsymbol{\Psi}\mathbf{T}t} + (\mathbf{I}_m-\mathbf{S}_m ) e^{\mathbf{T}t} \Big) \mathbf{T}^{-1}\mathbf{D}\mathbf{e}_j.
\end{equation}

Following (\ref{eq:alpha2}) and (\ref{eq:sumF}), the instantaneous intensity $\alpha(t)$ is given by
\begin{cor}
For any $j\in\{m+1,...,m+p\}$, we have for all $t\geq 0$
\begin{equation}\label{eq:intensities3}
\begin{split}
\alpha_{j}(t)=&\frac{ \boldsymbol{\pi}^{\top}\big(\mathbf{S}_m e^{\boldsymbol{\Psi}\mathbf{T}t}\boldsymbol{\Psi}
+ \big(\mathbf{I}_m-\mathbf{S}_m\big) e^{\mathbf{T}t} \big)\mathbf{D} \mathbf{e}_j   }{ \boldsymbol{\pi}^{\top}\big(\mathbf{S}_m e^{\boldsymbol{\Psi}\mathbf{T}t}
+ \big(\mathbf{I}_m-\mathbf{S}_m \big) e^{\mathbf{T}t} \big) \mathbf{1}_m    } 
\end{split}
\end{equation}
\end{cor}

\begin{Rem}
By setting $\boldsymbol{\Psi}=\mathbf{I}$ the type-specific distribution $F_j(t)$ (\ref{eq:competing2}) and intensity $\alpha_j(t)$ (\ref{eq:intensities3}) are reduced to the same expression given in Lindqvist \cite{Lindqvist}:
\begin{equation*}
\begin{split}
F_j(t)=\boldsymbol{\pi}^{\top} \mathbf{T}^{-1}\big( e^{\mathbf{T}t} - \mathbf{I}_m\big) \mathbf{D} \mathbf{e}_j  \quad \mathrm{and} \quad
\alpha_j(t)= \frac{\boldsymbol{\pi}^{\top} e^{\mathbf{T} t} \mathbf{D} \mathbf{e}_j}{\boldsymbol{\pi}^{\top} e^{\mathbf{T}t} \mathbf{1}_m}.
\end{split}
\end{equation*}
\end{Rem}

Solving differential equation (\ref{eq:alpha2}) for $f_j(t)$ and $F_j(t)$ we can express similar to (\ref{eq:survival2}) and (\ref{eq:seeminglyindp}) the latter two in terms of $\alpha_j(t)$ and overall intensity $\alpha(t)$ as
\begin{equation*}
\begin{split}
f_{j}(t)=&\alpha_{j}(t)\exp\Big( -\int_0^t \sum_{j=1}^p \alpha_{j}(u) du\Big)\\
F_{j}(t)=& \int_0^t \alpha_{j}(u)\exp\Big( -\int_0^u \sum_{j=1}^p \alpha_{j}(v) dv\Big) du.
\end{split}
\end{equation*}

The distribution of age at the time $X$ gets absorbed by type $j$ and the probability of ultimate absorption by type $j$ are defined following \cite{Crowder} by
\begin{equation*}
\widetilde{F}_{j}(t)=\mathbb{P}\{\tau \leq t \vert \mathbf{J}=j\}
\quad \mathrm{and} \quad \widetilde{p}_{j}(t)=\mathbb{P}\{\mathbf{J}=j \vert \tau >t\}.
\end{equation*}
\begin{prop}
For any $j=m+1,...,m+p$, we have for all $t\geq 0$,
\begin{equation*}
\begin{split}
\widetilde{F}_{j}(t)=& -\frac{  \boldsymbol{\pi}^{\top}\big(\mathbf{S}_m (\boldsymbol{\Psi}\mathbf{T})^{-1}\big(e^{\boldsymbol{\Psi}\mathbf{T}t}-\mathbf{I}_m\big)\boldsymbol{\Psi}
+ \big(\mathbf{I}_m-\mathbf{S}_m\big) T^{-1}\big(e^{\mathbf{T}t}-\mathbf{I}_m\big)\big)\mathbf{D} \mathbf{e}_j  }
{  \boldsymbol{\pi}^{\top} \mathbf{T}^{-1} \mathbf{D} \mathbf{e}_j  } \\
&\hspace{0.75cm}\widetilde{p}_{j}(t)= -\frac{ \boldsymbol{\pi}^{\top}\big(\mathbf{S}_m e^{\boldsymbol{\Psi}\mathbf{T}t} + (\mathbf{I}_m-\mathbf{S}_m) e^{\mathbf{T}t} \Big) \mathbf{T}^{-1}\mathbf{D}\mathbf{e}_j   }{ \boldsymbol{\pi}^{\top}\Big(\mathbf{S}_m e^{\boldsymbol{\Psi}\mathbf{T}t}
+ \big(\mathbf{I}_m-\mathbf{S}_m \big) e^{\mathbf{T}t} \big) \mathbf{1}_m    }.
\end{split}
\end{equation*}
\end{prop}

\subsection{Residual lifetime under competing risks}

Following \cite{Crowder}, the residual lifetime for competing risk is defined similar to (\ref{eq:residuedef}):
\begin{equation*}
R_{ij}(t)=\int_t^{\infty} \overline{F}_{ij}(t,u) du.
\end{equation*}
Applying the sub-distribution (\ref{eq:ultimate2}) to the above, we have the following.
\begin{prop}
For any $i=1,..,m$ and $j=m+1,...,m+p$, we have $\forall t\geq 0$,
\begin{equation}
R_{ij}(t)=\mathbf{e}_i^{\top} \big(\mathbf{S}_m(t) (\boldsymbol{\Psi}\mathbf{T})^{-1} + (\mathbf{I}_m-\mathbf{S}_m(t))\mathbf{T}^{-1}\big) \mathbf{T}^{-1}\mathbf{D}\mathbf{e}_j.
\end{equation}
\end{prop}

In the section below, we discuss some numerical examples of the main results of Sections \ref{section:Main} and \ref{sec:main3}. The examples are given for multi absorbing states model.

\section{Numerical examples}\label{section:examples}

We consider marriage and divorce problem under the mixture model (\ref{eq:mixture}). The marriage breakup is due to divorce in one-absorbing state Markov mixture model and by being widowed in competing risks model. We are interested in the lifetime distribution of marriage until a breakup due to divorce is observed. In one-absorbing state model, the transient states is $E=\{N, M, S, W\}$ whilst the absorbing states $\Delta=\{D\}$. We denote by $N$ for never married, $M$ for married, $W$ for widowed, $S$ for separated, and $D$ for divorced. For simplicity, we label $N=1$, $M=2$, $S=3$, $W=4$ and $D=5$. The transition rates from never married to married is given by $q_{12}$, married to widowed by $q_{24}$ and widowed to married by $q_{42}$. We assume that married couples may get separated after separation which occurs at rate of $q_{23}$, and eventually getting divorced or being widowed at the rate of $q_{35}$ or $q_{34}$, respectively. We allow the possibility that divorced may take place at rate $q_{25}$ straight from married. In the competing risks, we set $q_{42}=0$ and move the widowed state $W$ to the absorbing states, i.e., $\Delta=\{W,D\}$.


%
\begin{figure}[ht]
\begin{center}
\includegraphics[width=1\textwidth]{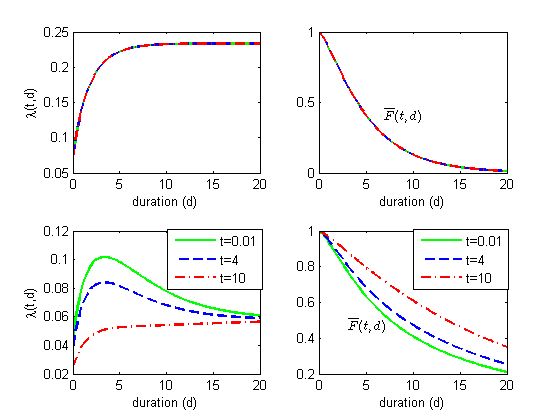}
\caption{Forward divorced intensity $\lambda_i(t,d)$ (\ref{eq:intensity2}) and survival distribution $\overline{F}_i(t,d)$ (\ref{eq:main}) as functions of marriage duration $d=s-t$ for a given age $t>0$. The two graphs on top correspond to the case where the population is homogeneous, i.e., $\boldsymbol{\Psi}=\mathbf{I}$, whereas the other two below for the case where there is heterogeneity in population; two different groups moving at different speed, i.e., $\boldsymbol{\Psi}\neq\mathbf{I}$.}
\label{fig:divorce}
\end{center}
\end{figure}

The composition of individuals in each state is denoted by initial vector $\boldsymbol{\pi}$.
The corresponding intensity matrix $\mathbf{Q}$ is given by (\ref{eq:ExMatT2}). Heterogeneity of individuals are described by the speed at which they move along the states. We assume that there are two groups of individuals: one moving with intensity matrix $\mathbf{T}$, whilst the other with $\boldsymbol{\Psi}\mathbf{T}$, with $\boldsymbol{\Psi}$ given by (\ref{eq:Lambda}). At initial time, the composition of these groups in each state is given by diagonal matrix $\mathbf{S}_m$ (\ref{eq:Snt}).
\begin{equation}\label{eq:ExMatT2}
\mathbf{Q} = \left(\begin{array}{ccccc}
  -q_{12} & q_{12} & 0  & 0 & 0 \\
  0 & -(q_{23}+q_{24}+q_{25}) & q_{23} & q_{24} & q_{25} \\
  0 & 0 & -(q_{34}+q_{35}) & q_{34} & q_{35} \\
  0 &  q_{42} &   0 & -q_{42} & 0 \\
  0 & 0 & 0 & 0 & 0
\end{array}\right).
\end{equation}
\begin{figure}[ht]
\begin{center}
\includegraphics[width=1.0\textwidth]{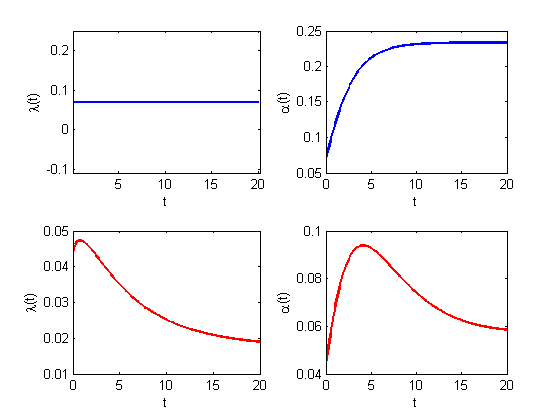}
\caption{Intensity rate $\lambda_i(t)$ (\ref{eq:intensity3}) and baseline function $\alpha(t)$ (\ref{eq:baseline}) as a function of the marriage age $t$. The two graphs on top correspond to the case where the population is homogeneous, i.e., $\boldsymbol{\Psi}=\mathbf{I}$, whereas the other two below for the case where there is heterogeneity; two groups moving at different speed, i.e., $\boldsymbol{\Psi}\neq\mathbf{I}$.}
\label{fig:divorce2}
\end{center}
\end{figure}
\begin{figure}[ht]
\begin{center}
\includegraphics[width=0.85\textwidth]{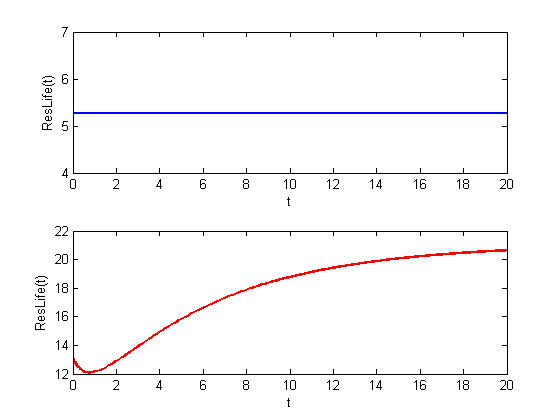}
\caption{Residual lifetime (\ref{eq:residue}) of the marriage against the marriage age $t$. The graph on top correspond to the case where the population is homogeneous, i.e., $\boldsymbol{\Psi}=\mathbf{I}$, whereas the one below is for the case where there is heterogeneity - two groups moving at different speed, i.e., $\boldsymbol{\Psi}\neq\mathbf{I}$ - and has U-bend shape.}
\label{fig:residue}
\end{center}
\end{figure}
\begin{figure}[ht]
\begin{center}
\includegraphics[width=1\textwidth]{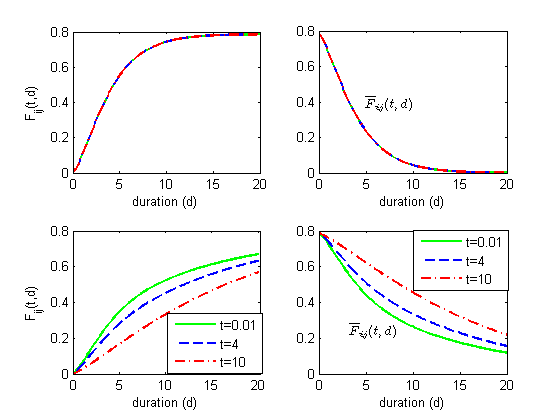}
\caption{Sub-distributions $F_{ij}(t,s)$ and $\overline{F}_{ij}(t,s)$ under competing risks. }
\label{fig:competing}
\end{center}
\end{figure}

For the purpose of simulation, we set the following values for the mixture model. We set the elements of intensity matrix $\mathbf{T}$ by $q_{12}=0.95$, $q_{24}=0.05$, $q_{34}=0.1$, $q_{23}=0.25$, $q_{25}=0.07$, $q_{42}=0.85$, and $q_{35}=0.5.$ Heterogeneity of the population is described by matrix $\mathbf{S}_m=\textrm{diag}(0.5,0.5,0.5,0.5,0.5)$, i.e., there are $50\%$ of the population in each state moving with intensity matrix $\mathbf{T}$, while the other $50\%$ moving with rate $\boldsymbol{\Psi}\mathbf{T}$ with $\boldsymbol{\Psi}=\textrm{diag}(0.25,0.25,0.25)$. The initial distribution of population is given by $\boldsymbol{\pi}=(0.5,0.3,0.1)^{\top}$; $50\%$ never married, $30\%$ married, $10\%$ widowed and the rest is being in separated status. Based on these values, we compute and compare intensity rate $\lambda_i(t,s)$ (\ref{eq:intensity2}), $\lambda_i(t)$ (\ref{eq:intensity3}) and baseline function $\alpha(t)$ (\ref{eq:baseline}) of getting divorced for individuals who have been married ($\mathbf{e}_i=(0,1,0,0)^{\top}$) for $t$ period of time. Heterogeneity is removed by setting $\boldsymbol{\Psi}$ to be identity matrix. We assume that we have only the current status of marriage (limited information on $\mathcal{I}_{i,t}$). The results are presented in Figures \ref{fig:divorce} and \ref{fig:divorce2}. We observe in case there is no heterogeneity in the population that the divorced intensity is always increasing in duration, regardless of the age of marriage $t$. Closure look at Figure \ref{fig:divorce2} suggests that the forward divorced intensity $\lambda(t,d)$ (\ref{eq:intensity2}), with $d=s-t$, is the same as the baseline function $\alpha(t)$ (\ref{eq:baseline}). Moreover, the instantaneous divorced intensity $\lambda(t)$ defined in (\ref{eq:intensity3}) is not affected by the marriage age $t$, i.e., divorced intensity $\lambda(t)$ is just a constant.

In contrast to this observation for heterogeneous population, we notice from Figure \ref{fig:divorce} that the divorced intensity $\lambda(t,d)$ (\ref{eq:intensity2}) changes its shapes w.r.t the marriage age $t$\footnote{It exhibits similar type of behavior to Fig 9 in \cite{Aalen1995} for different initial states.}. As we see from Figure \ref{fig:divorce} that those who have been married for $t=10$ period of time have the lowest divorced intensity and those who just got married for $t=0.01$ period of time have the highest divorced intensity of all. In between is the divorced intensity for those who have been married for $t=4.$

From these curves we notice that the newly married and have been married for few periods face increasing divorced intensity for the next 3.6 time period of their marriage. Especially at the beginning, the increment is rather steep. After approximately 3.6 period of time, the marriage is getting stable for which the divorced rate decreases to the long-term rate. Even though the divorced rate is relative low for old couples, it tends to increase at faster rate at the beginning and then gradually reaching up the long-term rate $0.0644$ (\ref{eq:limit2}). The instantaneous divorced intensity $\lambda(t)$ (\ref{eq:intensity3}) confirms the above phenomenon that the divorced intensity tends to increase for newly married, and then the relationship gets more stable as the age of marriage $t$ increases as evidenced by the decreasing rate of divorce. The baseline curve $\alpha(t)$ in Figure \ref{fig:divorce2} also exhibits similar observations.

In terms of survival probability, we see from Figure \ref{fig:divorce} that in both models (with or without heterogeneity) the marriage survival probability decreases as the marriage duration time $d$ increases. This is due to the fact that both intensity matrices $\mathbf{T}$ and $\boldsymbol{\Psi}\mathbf{T}$ are negative definite. The only main differences we observe is that the survival probability goes to zero faster under homogeneity model than under heterogeneity as seen by the strict increasing intensity for the existing model, see Figure \ref{fig:divorce}. This is due to $\psi_i<1$ for all $i=1,...,3$. As a result, the model implies that every marriage gets divorced in longer run, something which does not necessarily occur in practice. In contrast to this observation, we see under the mixture model (with heterogeneity) that the survival probability does not decay to zero at the end of observation period; marriage would last in longer run with a positive probability. Moreover, we can investigate the effect of the marriage age $t$ to the marriage survival probability. We observe that for the same marriage duration, the survival probability increases by the age, i.e., the older the age of marriage, the higher the probability that the marriage would survive.

Even though the above observations are based on numerical study on married-and-divorced model, but the results from the proposed model are able to provide sound explanations to the data depicted in Figure 5.3 on page 222 of \cite{Aalen2008}.

Figure \ref{fig:residue} depicts expected residual lifetime of marriage. We observe that the residual lifetime under homogeneity (the existing model) is not affected by the age of marriage $t$, i.e., it is just flat. In contrast to this observation, we notice that the residual lifetime exhibits a U-bend shape curve under heterogeneity (the new model), i.e., the residual lifetime is decreasing for newly married, and then increasing as the age of marriage gets older. These numerical findings seem to be consistent with U-bend shape of lifetime discussed in the 2010 \textit{The Economist} article. Thus, once again, we see that the proposed phase-type distribution and its forward intensity offer significant improvements over the existing model.


Figure \ref{fig:competing} displays the sub-distribution functions $F_{ij}(t,s)$ and $\overline{F}_{ij}(t,s)$ (\ref{eq:competing}) under competing risks in which case $q_{42}$ was set to zero. We again observe that under the Markov mixture model we can study the effect of heterogeneity to divorce probability. The figure exhibits higher conditional forward probability of getting divorce for newly marriage and lower for older marriage. This observation is in line with that of for conditional forward divorce intensity given in Figure \ref{fig:divorce}.

\section{Conclusions}

We have presented in this paper a generalization to the phase-type distribution introduced in Neuts \cite{Neuts1975} and \cite{Neuts1981}. The generalization is made to allow the inclusion of available past information of underlying process and for modeling heterogeneity. The new distribution and its dependence on the past information are given in an explicit form. It is constructed by extending the mixture model \cite{Frydman2005} and \cite{Frydman2008} into a mixture of two continuous-time finite-state absorbing Markov chains moving on the same state space at different speeds. The distribution has dense and closure properties over convex mixtures and finite convolutions.

Its availability in closed form would give an advantage of getting some analytically tractable results in applications. We have also proposed forward intensity associated with the new distribution. This intensity is responsible for determining rate of occurrence of future events based on available past information of the mixture process. We extended the phase-type model for multi-absorbing states which can be used to deal with competing risks in survival analysis. Numerical study on married-and-divorced problem was performed to compare the performance of the new distribution and its intensity against their existing counterparts. Following the results, we have seen that the proposed models provide significant improvements over the existing models. In particular, the forward intensity was able to provide sound explanations to divorced rate data found on page 222 of Aalen et al. \cite{Aalen2008}, and the residual lifetime reveals the common U-bend shape of lifetime found, e.g., in the 2010 \textit{The Economist} article.

Given its explicit form and the ability to capture path dependence and heterogeneity, we believe that the proposed phase-type distribution and forward intensity should be able to offer appealing features for wide range of applications which the existing distribution and intensity have been widely applied to.

\section*{Acknowledgement}
B.A. Surya acknowledges the support and hospitality provided by the
Center for Applied Probability of Columbia University during his academic visit in May 2014 where he saw the work of Frydman and Schuermann \cite{Frydman2008} and started to work on the problem. He thanks Professor Karl Sigman for the invitation. This research is financially supported by Victoria University PBRF Research Grants \# 212885 and \# 214168 for which the author would like to acknowledge.

\end{document}